\documentclass[12pt]{article}

\usepackage{amsthm}
\usepackage{amsfonts}
\usepackage{amssymb}
\usepackage{amsmath}
\usepackage{color}

\newcommand{\D}{\mathfrak {D}}

\def\pjpm{\psi_{\pm, (j,m_j,\kappa_j)}}

\def\psipm{\psi_{\pm}}
\def\psimp{\psi_{\mp}}

\def\1{1\hskip-0.09cm{\rm l}} %identity operator

\def\ael{a_1}
\def\ap{a_2}

\newcommand{\Ne}{N_+}
\newcommand{\Np} {N_-}
\newcommand{\Nei} {N_{+,\infty}}
\newcommand{\Nez} {N_{+,0}}
\newcommand{\Npi} {N_{-,\infty}}
\newcommand{\Npz} {N_{-,0}}

\newcommand{\R}{{\mathbb R}} % Reels
\newcommand{\C}{{\mathbb C}} % Complexes
\newcommand{\N}{{\mathbb N}} % Entiers naturels
\newcommand{\Z}{{\mathbb Z}} % Entiers relatifs
\newcommand{\gH}{{\mathfrak {H}}} %total Hilbert space
\newcommand{\gS}{{\mathfrak {S}}} % espace de Hilbert une particule bosonique
\newcommand{\gF}{{\mathfrak {F}}}

\def\gam{\gamma}
\def\p{\mathbf{p}}

\def\d{{\rm d}}

\def\z0{Z^0}

\def\w{\omega}

\def\me{\mathrm{m_e}} %masse de l'electron et du positron. notation speciale pour le distinguer des nombres quantiques m_j

\def\mz{\mathrm{m_{Z^0}}} %masse du boson Z

\newtheorem{prop}{Proposition}[section]
\newtheorem{corr}[prop]{Corollary}
\newtheorem{theo}[prop]{Theorem}
\newtheorem{lem}[prop]{Lemma}
\newtheorem{defi}[prop]{Definition}
\newtheorem{rema}[prop]{Remark}
\newtheorem{hypo}[prop]{Hypothesis}

\begin{document}

\title{Spectral Theory near Thresholds for Weak Interactions with Massive Particles}
\author{
%%%%%%%%%%%%%%%%%%%%%%%%%%%%%
{\bf Jean-Marie Barbaroux
 \footnote{E-mail: barbarou@univ-tln.fr}}\\
 {\small \it  Aix-Marseille Universit\'e, CNRS, CPT, UMR 7332, 13288 Marseille, France}\\
 {\small \it et Universit\'e de Toulon, CNRS, CPT, UMR 7332, 83957 La Garde, France}\\
 \and
%%%%%%%%%%%%%%%%%%%%%%%%%%%%%
{\bf J\'er\'emy Faupin
 \footnote{E-mail: jeremy.faupin@univ-lorraine.fr}}\\
 {\small \it Institut Elie Cartan de Lorraine, Universit{\'e} de Lorraine,}\\
 {\small \it 57045 Metz Cedex 1, France}\\
 \and
%%%%%%%%%%%%%%%%%%%%%%%%%%%%%
{\bf Jean-Claude Guillot
 \footnote{E-mail: guillot@cmapx.polytechnique.fr}}\\
 {\small \it CNRS-UMR 7641, Centre de Math\'ematiques Appliqu\'ees, Ecole Polytechnique}\\
 {\small \it 91128 Palaiseau Cedex, France} }

%%%%%%%%%%%%%%
\date{\today}

%%%%%%%%%%%%%%
\maketitle

%%%%%%%%%%%%%%%%%%%%%%%%%%%%%
\begin{abstract}
We consider a Hamiltonian describing the weak decay of the massive vector boson $Z^0$ into electrons and positrons. We show that the spectrum of the Hamiltonian is composed of a unique isolated ground state and a semi-axis of essential spectrum. Using a suitable extension of Mourre's theory, we prove that the essential spectrum below the boson mass is purely absolutely continuous.
\end{abstract}
%%%%%%%%%%%%%%%%%%%%%%%%%%%%%

\maketitle

%%%%%%%%%%%%%%%%%%%%%%%%%%%%%%%%%%%%%%%%%%%%%%%%%%%%%%%%%%%%%%%%%%%%%%%%%
%%%%%%%%%%%%%%%%  INTRODUCTION %%%%%%%%%%%%%%%%%%%%%%%%%%%%%%%%%%%%%%%%%%
%%%%%%%%%%%%%%%%%%%%%%%%%%%%%%%%%%%%%%%%%%%%%%%%%%%%%%%%%%%%%%%%%%%%%%%%%

\section{Introduction}\label{S1}

In this paper, we study a mathematical model for the weak decay of the vector boson $Z^0$ into electrons and positrons. The model we consider is an example of models of the weak interaction that can be patterned according to the Standard Model of Quantum Field Theory. Another example, describing the weak decay of the intermediate vector bosons $W^{\pm}$ into the full family of leptons, has been considered previously in \cite{bg4,ABFG}. Comparable models describing quantum electrodynamics processes can be constructed in a similar manner, see \cite{BDG}. We also mention \cite{Georgescu,GP} where the spectral analysis of some related abstract quantum field theory models have been studied.

Unlike \cite{ABFG}, the physical phenomenon considered in the present paper only involves \emph{massive} particles. In some respects, e.g. as far as the existence of a ground state is concerned, this feature considerably simplifies the spectral analysis of the Hamiltonian associated with the physical system we study. The main drawback is that, due to the positive masses of the particles, an infinite number of \emph{thresholds} occur in the spectrum of the free Hamiltonian (i.e. the full Hamiltonian where the interaction between the different particles has been turned off). Understanding the nature of the spectrum of the full Hamiltonian near the thresholds as the interaction is turned on then becomes a subtle question. Spectral analysis near thresholds, in particular by means of perturbation theory, is indeed well-known to be a delicate subject. This is the main concern of the present work.

Our main result will provide a complete description of the spectrum of the Hamiltonian below the boson mass. We will show that the spectrum is composed of a unique isolated eigenvalue $E$ (the ground state energy), and the semi-axis of essential spectrum $[E+\me , \infty)$, $\me$ being the electron mass. Moreover, using a version of Mourre's theory allowing for a non self-adjoint conjugate operator and requiring only low regularity of the Hamiltonian with respect to this conjugate operator, we will prove that the essential spectrum below the boson mass is purely absolutely continuous.

Before precisely stating our main results in Section \ref{section:results}, we begin with introducing in details the physical model we consider.

\section{Description of the model}\label{S2}

\subsection{The Fock space of electrons, positrons and $\z0$ bosons}

\subsubsection{Free Dirac operator}\label{S2.1.1}

The energy of a free relativistic electron of mass $\me$  is described by the Dirac Hamiltonian
(see \cite{ref9, ref8} and references therein)
\begin{equation}\nonumber
H_D := \boldsymbol{\alpha}
\cdot \frac{1}{i}\nabla + \boldsymbol{\beta}\, \me ,
\end{equation}
acting on the Hilbert space $\gH=L^2(\R ^3;\C^4)$, with domain $\D(H_D) =
H^1(\R ^3;\C^4)$. We use a system of units such that $\hbar = c =1$.
Here $\boldsymbol{\alpha} = (\alpha_1,\alpha_2,\alpha_3)$ and $\boldsymbol{\beta}$ are the
Dirac matrices in the standard form:
 $$ \boldsymbol{\beta} =
   \left(
    \begin{array}{cc}
            I & 0\\
                        0 &-I
    \end{array}
   \right), \qquad
   \alpha_i =
   \left(
    \begin{array}{cc}
            0 & \sigma_i\\
            \sigma_i & 0
    \end{array}
   \right), \qquad i = 1, 2, 3,
 $$
where $\sigma_i$ are the usual Pauli matrices. The operator $H_D   $ is self-adjoint, and $\mathrm{spec}(H_D) = (-\infty,\, - \me]\cup[ \me,\, +\infty)$.

The generalized eigenfunctions associated with the continuous spectrum of the
Dirac operator $H_D$ are labeled by the total angular momentum
quantum numbers
\begin{equation}\label{eq:def-qn1}
    j \in \big \{ \frac{1}{2}, \frac{3}{2}, \frac{5}{2},
    \ldots \big \}, \quad \  m_j \in \{ -j, -j+1, \ldots, j-1, j \},
\end{equation}
 and by the quantum numbers
\begin{equation}\label{eq:def-qn2}
    \kappa_j \in \big \{ \pm ( j+ \frac{1}{2}) \big \} .
\end{equation}

In the sequel, we will drop the index $j$ and set
\begin{equation}\label{eq:def-qn3}
 \gam=(j,\, m_j,\, \kappa_j)\, ,
\end{equation}
and a sum over $\gamma$ will thus denote a sum over $j\in\N+\frac12$, $m_j\in\{  -j, -j+1, \ldots, j-1, j \}$ and $\kappa_j\in\{ \pm ( j+ \frac{1}{2})\}$. We denote by $\Gamma$ the set $\{ ( j , \, m_j , \, \kappa_j ) , \, j\in\N+\frac12, \, m_j\in\{  -j, -j+1, \ldots, j-1, j \} , \, \kappa_j\in\{ \pm ( j+ \frac{1}{2}) \} \}$.

For $\p\in\R ^3$ being the momentum of the electron, and $p:=|\p|$, the continuum energy levels are given by $\pm\,  \w(p) $, where
\begin{equation}\label{eq:3}
  \w(p)  :=(\me^2 + p^2)^{\frac{1}{2}}\, .
\end{equation}
We set  the notation
\begin{equation}\label{eq:def-qn4}
 \xi = (p,\, \gamma)\, \in \mathbb{R}_+ \times \Gamma .
\end{equation}

The continuum eigenstates of $H_D$ are denoted by (see Appendix~\ref{appendixA} for a detailed description)
 $$
   \psipm(\xi, x) = \psipm((p,\gamma), x)\, .
 $$
We then have
 $$
  H_D\
  \psipm((p,\gamma), x) = \pm\,  \w(p)\  \psipm((p,\gamma), x) .
 $$

The generalized eigenstates $\psipm$ are here normalized in such a way that
\begin{eqnarray*}
 \int_{\R ^3}
    \psipm^\dagger((p,\gamma), x) \, \psipm((p',\gamma'), x)\, \d  x
    &=&
    \delta_{\gamma\gamma'}\delta(p-p'),\\
    \int_{\R ^3}
    \psipm^\dagger((p,\gamma), x)\, \psimp((p',\gamma'), x)\, \d x
    &=&0\, .
\end{eqnarray*}
Here $\psipm^\dagger((p,\gamma),\, x)$ is the adjoint spinor of $\psipm((p,\gamma),\,x)$.

\medskip

According to the hole theory \cite{IZ, ref9, Schweber,
ref8, W}, the absence in the Dirac theory of an electron with
energy $E < 0$ and charge $e$ is equivalent to the presence of a
positron with energy $-E > 0$ and charge $-e$.

Let us split the Hilbert space ${\mathfrak H} = L^2(\R ^3 ; \C^4)$
into
 $$
  {\mathfrak H}_{c^-} = P_{(-\infty,-\me]}(H_D){\mathfrak H}\quad \mbox{and}
  \qquad {\mathfrak H}_{c^+} = P_ {[\me,
  +\infty)}(H_D){\mathfrak H}.
 $$
Here $P_I(H_D)$ denotes the spectral projection of $H_D$
corresponding to the interval $I$.

Let $\Sigma := \mathbb{R}_+ \times \Gamma$. We can identify the Hilbert spaces ${\mathfrak H}_{c^\pm}$ with
 $$
\mathfrak{H}_{c} := L^2( \Sigma;\C^4) \simeq \oplus_\gamma L^2(\R _+; \C^4)\, ,
 $$
by using the unitary operators
$U_{c^\pm}$ defined from ${\mathfrak H}_{c^\pm}$ to ${\mathfrak H}_{c}$ as
\begin{equation}\label{eq:def-unitary-gamma}
 (U_{c^\pm}\phi) ( p , \gamma ) =  {\rm L.i.m}\, \int\, \psipm^\dagger((p,\gamma)\, ,x)\, \phi(x)\, \d x \, .
\end{equation}

On $\mathfrak{H}_{c}$, we define the scalar products
\begin{equation}\label{eq:notation-int-dxi}
 (g,\, h) = \int\overline{g(\xi)} h(\xi) \d \xi =
 \sum_{ \gamma \in \Gamma } \int_{\R^+} \overline{g(p, \gamma)} h(p, \gamma) \, \d p\, .
\end{equation}
In the sequel, we shall denote the variable $(p,\, \gamma)$ by $\xi_1 = ( p_1 , \gamma_1 )$ in the case of electrons, and
$\xi_2 = ( p_2 , \gamma_2 )$ in the case of positrons, respectively.

\subsubsection{The Fock space for electrons and positrons}

Let
\begin{equation*}
{\mathfrak F}_{a}: = {\mathfrak F}_{a}(\gH_c) = \bigoplus_{n=0}^\infty \, \otimes_a^n \gH_c  ,
\end{equation*}
be the Fermi-Fock space over  $\mathfrak{
H}_{c}$,
and let
 $$
  {\mathfrak F}_D := {\mathfrak F}_{a} \otimes  {\mathfrak F}_{a}
 $$
be the Fermi-Fock space for electrons and positrons, with vacuum $\Omega_D$ (see Appendix~\ref{appendixC} for details).

%%%%%%%%%%%%%%%%%%%%%%%%%%%
%%%%%%%%%%%%%%%%%%%%%%%%%%%
%%%%%%%%%%%%%%%%%%%%%%%%%%%
\subsubsection{Creation and annihilation operators for electrons and positrons}

We set, for every $g \in \mathfrak H_{c}$,
\begin{eqnarray*}
b_{\gamma,+}(g)   &=& b_+  (P_\gamma^+ g)\, ,\\ \noalign{\vskip 8pt}
b_{\gamma,+}^*(g)&=& b_+^*(P_\gamma^+ g)\, ,
\end{eqnarray*}
where $P_\gamma^+$ is the projection of $\mathfrak H_{c}$ onto the
$\gamma$-th component defined according to \eqref{eq:def-unitary-gamma}, and $b_+  (P_\gamma^+ g)$ and $b_+^*  (P_\gamma^+ g)$ are respectively the annihilation and creation operator for an electron defined in Appendix~\ref{appendixC}.

As above, we set, for every $h \in \mathfrak H_{c}$,
\begin{eqnarray*}
b_{\gamma,-}(h)   &=& b_{-}  (P_\gamma^- h)\, ,\\ \noalign{\vskip 8pt}
b_{\gamma,-}^*(h) &=& b_-^*(P_\gamma^- h)\, ,
\end{eqnarray*}
where $P_\gamma^-$ is the projection of $\mathfrak H_{c}$ onto the
$\gamma$-th component, and $b_-  (P_\gamma^- g)$ and $b_-^*  (P_\gamma^- g)$ are respectively the annihilation and creation operator for a positron defined in Appendix~\ref{appendixC}.

As in \cite[Chapter X]{RS}, we introduce
operator-valued distributions $b_{\pm}(\xi)$ and $b_{\pm}^*(\xi)$ that fulfill for $g\in\mathfrak{H}_c$,
\begin{equation}\nonumber
\begin{split}
 & b_{\pm}(g)   = \int\
 b_{\pm}(\xi)\,  \overline{(P_\gamma^\pm g)\, (p)} \, \d \xi\\
 & b_{\pm}^*(g) = \int \ b_{\gamma,\pm}^*(p)\,
 (P_\gamma^\pm g)\, (p)\, \d \xi
\end{split}
\end{equation}
where we used the notation of \eqref{eq:notation-int-dxi}.
%%%%%%%%%%%%%%%%%%%%%%%%%%%
%%%%%%%%%%%%%%%%%%%%%%%%%%%
%%%%%%%%%%%%%%%%%%%%%%%%%%%
%
%

\subsubsection{Fock space for the $\z0$ boson.}\label{S-boson}

Let $\gS$ be any separable Hilbert space. Let $\otimes _s^n\gS$
denote the symmetric $n$-th tensor power of $\gS$. The symmetric Fock space over $\gS$, denoted by
$\gF_s(\gS)$, is the direct sum
\begin{equation}\label{2.3}
   \gF_s(\gS) = \bigoplus_{n=0}^\infty \otimes_s^n \gS\ ,
\end{equation}
where $\otimes_s^0\gS \equiv\C$. The state
$\Omega_s = (1,0,0,\ldots,0,\ldots)$ denotes the vacuum state in $\gF_s(\gS)$.

Let
$$
 \Sigma_3 := \R^3\times\{ -1,\, 0,\, 1\}\, .
$$
The one-particle Hilbert space for the particle $\z0$ is $L^2(\Sigma_3)$ with scalar product
\begin{equation}
 (f,\, g) = \int_{\Sigma_3} \overline{f(\xi_3)} g(\xi_3) \d\xi_3\, ,
\end{equation}
with the notations
\begin{equation}\label{eq:notation-boson}
 \xi_3 = (k,\, \lambda)\quad\mbox{and}\quad\int_{\Sigma_3} \d \xi_3 = \sum_{\lambda=-1,0,1} \int_{\R^3} \d k  \, ,
\end{equation}
where $\xi_3 = (k,\, \lambda)\in\Sigma_3$.

The bosonic Fock space for the vector boson $\z0$,
denoted by $\gF_{\z0}$, is thus
\begin{equation}\label{2.8}
 \gF_{\z0} = \gF_s( L^2(\Sigma_3) )\, .
\end{equation}

For $f\in L^2( \Sigma_3 )$, we define the annihilation and creation operators, denoted by $a(f)$ and $a^*(f)$ by
\begin{equation}
 a(f) = \int_{\Sigma_3} \overline{f(\xi_3)} a(\xi_3) \d\xi_3
\end{equation}
and
\begin{equation}
 a^*(f) = \int_{\Sigma_3} {f(\xi_3)} a^*(\xi_3) \d\xi_3
\end{equation}
where the operators $a(\xi_3)$ (respectively $a^*(\xi_3)$) are the bosonic annihilation (respectively bosonic creation) operator for the boson $\z0$ (see e.g \cite{ref10, bg3, bg4}).

\subsection{The Hamiltonian}

\subsubsection{The free Hamiltonian}
The quantization of the Dirac Hamiltonian $H_D$, denoted by
$\mathrm{d}\Gamma(H_D)$, and acting on $\mathfrak{F}_D$, is given by
\begin{equation}\nonumber
\mathrm{d}\Gamma(H_D) =
 \int \omega(p)\, b_{+}^*(\xi_1)\, b_{+}(\xi_1) \d \xi_1\  + \  \int  \omega(p)\,
 b_{-}^*(\xi_2)\, b_{-}(\xi_2) \d \xi_2 ,
\end{equation}
with $\omega(p)$ given in \eqref{eq:3}.
The operator $\d\Gamma(H_D)$ is the
Hamiltonian of the quantized Dirac field.

Let $\D_D$ denote the set of vectors $\Phi \in {\mathfrak F}_D$ for which $\Phi^{(r,s)}$ is smooth and has a compact  support and $\Phi^{(r,s)} = 0$ for all but finitely many $(r,s)$. Then $\d\Gamma( H_D )$ is well-defined on the
dense subset $\D_D$ and it is essentially self-adjoint on $\D_D$. The
self-adjoint extension will be denoted by the same symbol $\d\Gamma(H_D)$, with domain
$\D(\d\Gamma(H_D))$.

The operators number of electrons and number of positrons, denoted respectively by $\Ne$ and $\Np$, are given by
 \begin{equation}\label{eq:def-NpNe}
  \Ne = \int \, b_{+}^*(\xi_1)\, b_{+}(\xi_1) \d \xi_1
  \quad\mbox{and}\quad
  \Np = \int \, b_{-}^*(\xi_2)\, b_{-}(\xi_2) \d \xi_2\, .
 \end{equation}
They are essentially self-adjoint on $\D_D$. Their self-adjoint extensions will be also denoted by $\Ne$ and $\Np$.

We have
 $$
  \mathrm{spec}(\d\Gamma(H_D)) = \{0\} \cup [\me , \infty) .
 $$
The set $[\me, \infty)$ is the absolutely continuous spectrum of $\d\Gamma(H_D)$.

The Hamiltonian of the bosonic field, denoted by $\d\Gamma( H_{Z^0} )$, acting on $\mathfrak{F}_{\z0}$, is
 $$
  \d\Gamma( H_{Z^0} ) := \int \omega_3(k)\, a^*(\xi_3) a(\xi_3) \, \d \xi_3
 $$
where
\begin{equation}\label{eq:omega3}
\omega_3(k) = \sqrt{|k|^2 + {\mz}^2}.
\end{equation}
The operator $\d\Gamma( H_{Z^0} )$ is essentially self-adjoint on the set of vectors $\Phi \in \mathfrak{F}_{\z0}$ such that $\Phi^{(n)}$ is smooth and has compact support and $\Phi^{(n)} = 0$ for all but finitely many $n$. Its self-adjoint extension is denoted by the same symbol.

The spectrum of $\d\Gamma( H_{Z^0} )$ consists of an absolutely continuous spectrum covering
$\lbrack \mz, \infty)$ and a simple eigenvalue, equal to zero, whose
corresponding eigenvector is  the vacuum state  $\Omega_s \in {\mathfrak
F}_{\z0}$.

The free Hamiltonian is defined on $\mathcal{H} := \mathfrak{F}_D \otimes \mathfrak{F}_{\z0}$ by
\begin{equation}
  H_0 = \d \Gamma(H_D)\otimes \1 + \1 \otimes\d \Gamma( H_{Z^0} )\, .
\end{equation}
The operator $H_0$ is essentially self-adjoint on $\D(\d\Gamma(H_D))\otimes\D(\d\Gamma( H_{Z^0} ))$.
Since $\me < \mz$, the spectrum of $H_0$ is given by
$$
 \mathrm{spec}(H_0) = \{0\} \cup [\me,\, \infty)\, .
$$
More precisely,
\begin{equation}\label{eq:spec_H0}
 \mathrm{spec}_{ \mathrm{pp} }(H_0) = \{0\}, \quad  \mathrm{spec}_{ \mathrm{sc} }(H_0) = \emptyset, \quad  \mathrm{spec}_{ \mathrm{ac} }(H_0) = [\me,\, \infty ),
\end{equation}
where $ \mathrm{spec}_{ \mathrm{pp} }$, $\mathrm{spec}_{ \mathrm{sc} }$, $\mathrm{spec}_{ \mathrm{ac} }$ denote the pure point, singular continuous and absolutely continuous spectra, respectively. Furthermore, $0$ is a non-degenerate eigenvalue associated to the vacuum $\Omega_D \otimes \Omega_s$.

\subsubsection{The Interaction}\label{S-Interaction}

The interaction between the electrons/positrons and the boson vectors $Z^0$, in the Schr\"odinger
representation, is given, up to coupling contant, by (see
\cite[(4.139)]{GreinerMuller} and
\cite[(21.3.20)]{W2})
\begin{equation}\label{def:formal-interaction}
I  =  \int  \overline{\Psi_e}(x) \gamma^\alpha
  (g'_V-\gamma_5)\Psi_{e}(x) Z_\alpha(x) \, \d x\, + h.c. ,
\end{equation}
where $\gamma^\alpha$, $\alpha=0,1,2,3$, and $\gamma_5$ are the
Dirac matrices, $g'_V$ is a real parameter such that $g'_V \simeq 0,074$ (see e.g \cite{GreinerMuller}), $\Psi_{e}(x)$ and $\overline{\Psi_{e}}(x)$ are the
Dirac fields for the electron $e_-$ and the positron $e_+$ of mass $\me$, and $Z_\alpha$
is the massive boson field for $Z^0$.

With the notations of Subsection \ref{S2.1.1},  $\Psi_e(x)$ is formally defined by
\begin{equation}\nonumber
\begin{split}
 \Psi_e(x) =  \int \psi_+(\xi,\, x) b_+(\xi)
 + \widetilde{\psi}_-(\xi,\, x) b^*_-(\xi)\, \d \xi ,
\end{split}
\end{equation}
where
\begin{equation}\label{intro2}
 \widetilde{\psi}_-(\xi,\, x) = \widetilde{\psi}_-((p,\gamma),\, x)
 =  \psi_{-}((p,(j,-m_{j},-\kappa_{j})),x)\, .
\end{equation}

The boson field $Z_\alpha$ is formally defined by (see e.g. \cite[Eq.~(5.3.34)]{W}),
\begin{equation}\nonumber
\begin{split}
 & Z_\alpha(x) \\
 & = {(2\pi)}^{-\frac32}
 \int \frac{\d \xi_3}{(2 (|k|^2 \!+\! \mz^2)^\frac12)^\frac12}
  \Big(\epsilon_\alpha(k,\lambda)
 a(\xi_3)
 \mathrm{e}^{i k.x}
 + \epsilon_\alpha^*(k,\lambda) a^*(\xi_3)
 \mathrm{e}^{- i k.x}\Big)\, ,
 \end{split}
\end{equation}
with $\xi_3 = (k,\, \lambda)$ according to \eqref{eq:notation-boson}, and where the vectors
$\epsilon_\alpha(k,\lambda)$ are the polarizations vectors of the massive
spin~1 bosons (see \cite[Section~5.3]{W}).

If one considers the full interaction $I$ in \eqref{def:formal-interaction} describing the decay of the
gauge boson $Z^0$ into massive leptons and if one formally expands
this interaction with respect to products of creation and
annihilation operators, we are left with a finite sum of terms
with kernels yielding singular operators which cannot be defined as closed operators. Therefore,  in order to obtain a well-defined Hamiltonian (see e.g
\cite{GlimmJaffe, BDG, BDG2, bg4, ABFG}), we replace these kernels by square integrable functions $G^{(\alpha)}$.

This implies in particular to introduce cutoffs for high momenta of electrons, positrons and $Z^0$ bosons. Moreover, we confine in space the interaction between the electrons/positrons and the bosons by adding a localization
function $f(|x|)$, with $f \in \mathrm{C}_0^\infty( [ 0 , \infty ) )$.
The interaction Hamiltonian is thus defined on $\mathcal{H} = \mathfrak{F}_D\otimes\mathfrak{F}_{\z0}$ by
\begin{equation}\label{eq:interaction}
 H_I = H_I^{(1)} + {H_I^{(1)}}^* + H_I^{(2)} + {H_I^{(2)}}^*\, ,
\end{equation}
with
\begin{equation}\label{eq:interaction-1}
\begin{split}
 H_I^{(1)} & = \int
 \left(
  \int_{\R^3}  f(|x|) \overline{ \psi_+(\xi_1,\, x)} \gamma^\mu (g'_V - \gamma_5)\widetilde{ \psi}_-(\xi_2,\, x) \frac{\epsilon_\mu(\xi_3)}{\sqrt{2 \omega_{3}(k)}} \, \mathrm{e}^{i k \cdot x} \, \d x
 \right) \\
  & \ \ \ \times G^{(1)}(\xi_1,\, \xi_2,\, \xi_3) b_+^*(\xi_1) b_-^*(\xi_2) a(\xi_3)\, \d\xi_1\d\xi_2\d\xi_3\, ,
\end{split}
\end{equation}
\begin{equation}
\begin{split}
 {H_I^{(1)}}^* & = \int
 \left(
  \int_{\R^3}  f(|x|) \overline{\widetilde{\psi}_-(\xi_2,\, x)} \gamma^\mu (g'_V - \gamma_5) \psi_+(\xi_1,\, x)
   \frac{\epsilon_\mu^*(\xi_3)}{\sqrt{2 \omega_{3}(k)}} \, \mathrm{e}^{ -i k \cdot x} \, \d x
 \right) \\
  & \ \ \ \times \overline{G^{(1)}(\xi_1,\, \xi_2,\, \xi_3)} a^*(\xi_3) b_-(\xi_2) b_+(\xi_1) \, \d\xi_1\d\xi_2\d\xi_3\, ,
\end{split}
\end{equation}
\begin{equation}\label{eq:interaction-3}
\begin{split}
 H_I^{(2)} & = \int
 \left(
  \int_{\R^3}  f(|x|) \overline{\psi_+(\xi_1,\, x)} \gamma^\mu (g'_V - \gamma_5)\widetilde{ \psi}_-(\xi_2,\, x) \frac{\epsilon_\mu^{*}(\xi_3)}{\sqrt{2 \omega_{3}(k)}} \, \mathrm{e}^{-i k \cdot x} \, \d x
 \right) \\
  & \ \ \ \times G^{(2)}(\xi_1,\, \xi_2,\, \xi_3) b_+^*(\xi_1) b_-^*(\xi_2) a^*(\xi_3)\, \d\xi_1\d\xi_2\d\xi_3\, ,
\end{split}
\end{equation}
and
\begin{equation}\label{eq:interaction-4}
\begin{split}
 {H_I^{(2)}}^* & = \int
 \left(
  \int_{\R^3}  f(|x|) \overline{\widetilde{\psi}_-(\xi_2,\, x)} \gamma^\mu (g'_V - \gamma_5) \psi_+(\xi_1,\, x)
   \frac{\epsilon_\mu(\xi_3)}{\sqrt{2 \omega_{3}(k)}} \, \mathrm{e}^{ i k \cdot x} \, \d x
 \right) \\
  & \ \ \ \times \overline{G^{(2)}(\xi_1,\, \xi_2,\, \xi_3)} a(\xi_3) b_-(\xi_2) b_+(\xi_1) \, \d\xi_1\d\xi_2\d\xi_3\, .
\end{split}
\end{equation}
Performing the integration with respect to $x$ in the expressions above, we see that $H_I^{(1)}$ and $H_I^{(2)}$ can be written under the form
\begin{align}
  H_I^{(1)}  := H_I^{(1)}(F^{(1)}) & : = \int F^{(1)}(\xi_1,\, \xi_2,\, \xi_3) b_+^*(\xi_1) b_-^*(\xi_2) a(\xi_3)\, \d\xi_1\d\xi_2\d\xi_3\, , \label{eq:interaction-1b} \\
  H_I^{(2)}  := H_I^{(2)}(F^{(2)}) & : = \int F^{(2)}(\xi_1,\, \xi_2,\, \xi_3) b_+^*(\xi_1) b_-^*(\xi_2) a^*(\xi_3)\, \d\xi_1\d\xi_2\d\xi_3\, , \label{eq:interaction-2b}
\end{align}
where, for $\alpha=1,2$,
\begin{align}
 F^{(\alpha)}(\xi_1,\, \xi_2,\, \xi_3) := h^{(\alpha)}( \xi_1 , \xi_2 , \xi_3 ) G^{(\alpha)}( \xi_1 , \xi_2 , \xi_3 )   , \label{eq:interaction-12c}
\end{align}
and $h^{(1)}( \xi_1 , \xi_2 , \xi_3 )$, $h^{(2)}( \xi_1 , \xi_2 , \xi_3 )$ are given by the integral over $x$ in \eqref{eq:interaction-1} and \eqref{eq:interaction-3}, respectively.

Our main result, Theorem \ref{thm:abs-cont} below, requires the coupling functions $F^{(\alpha)}(\xi_1,\, \xi_2,\, \xi_3)$ to be sufficiently regular near  $p_1 = 0$ and $p_2 = 0$ (where, recall, $\xi_l = ( p_l , \gamma_l )$ for $l=1,2$). The behavior of the generalized eigenstates $\psi_+(\xi,\, x)$ and $\psi_-(\xi,\, x)$ near $\xi = 0$, and therefore the behavior of $h^{(\alpha)}( \xi_1 , \xi_2 , \xi_3 )$ near $p_1 = 0$ and $p_2 = 0$, will be analyzed in Appendix \ref{appendixA}.

\subsubsection{The total Hamiltonian}\label{S2.4}

\begin{defi}\label{def:hamiltonian}
The Hamiltonian of the decay of the boson $\z0$ into an electron and a po\-si\-tron  is
 $$
  H := H_0 + g H_I\, .
 $$
where $g$ is a real coupling constant.
\end{defi}

\section{Main results}\label{section:results}

For $p\in\R_+$, $j\in\{\frac12,\, \frac32,\, \cdots \}$, $\gamma=(j,\, m_j,\, \kappa_j)$ and $\gamma_j=j+\frac12$, we define
\begin{equation}\label{eq:def-apj}
\begin{split}
   A(\xi) = A(p,\gamma):= \frac{(2p)^{\gamma_j + 1} }{\Gamma(\gamma_j)} \left(\frac{\w(p)+\me}{\w(p)}\right)^\frac12
   \left(\int_0^\infty |f(r)| r^{2\gamma_j} (1+r^2)\d r\right)^\frac12 ,
\end{split}
\end{equation}
where $\Gamma$ denotes Euler's Gamma function, and $f \in \mathrm{C}_0^\infty( [ 0 , \infty ))$ is the localization function appearing in \eqref{eq:interaction-1}--\eqref{eq:interaction-4}. We make the following hypothesis on the kernels $G^{(\alpha)}$.

\begin{hypo}\label{hyp:1}
For $\alpha=1,2$,
\begin{equation}\label{eq:hyp1}
 \int A(\xi_1)^2 A(\xi_2)^2 (|k|^2 + \mz^2)^{\frac12} \left| G^{(\alpha)}(\xi_1,\, \xi_2,\, \xi_3)\right|^2 \d \xi_1 \d \xi_2 \d \xi_3 < \infty.
\end{equation}
\end{hypo}

\begin{rema} Note that up to universal constants, the functions $A(\xi)$ in \eqref{eq:def-apj} are upper bounds for the integrals with respect to $x$ that occur in \eqref{eq:interaction-1}. These bounds are derived using the inequality (see \cite[Eq.(5.3.23)-(5.3.25)]{W})
\begin{equation}\label{eq:def-C}
 \left| \frac{\epsilon_\mu(\xi_3)}{\sqrt{2 \omega_{3}(k)}}\right|
  \leq C_\mz( 1 + |k|^2 )^{\frac14}\, .
\end{equation}
\end{rema}

For $C_\mz$ being the constant defined in \eqref{eq:def-C}, and $C_Z=156\, C_\mz$, let us define
\begin{align}
 K_1(G^{(\alpha)})^2 & := C_Z{}^2
 \left( \int A(\xi_1)^2 A(\xi_2)^2\, |G^{(\alpha)}(\xi_1,\, \xi_2,\, \xi_3)|^2 \d \xi_1 \d \xi_2 \d \xi_3 \right) ,\notag \\
 K_2(G^{(\alpha)})^2 & := C_Z{}^2
 \left( \int A(\xi_1)^2 A(\xi_2)^2\, |G^{(\alpha)}(\xi_1,\, \xi_2,\, \xi_3)|^2
 (|k|^2 + 1)^\frac12 \d \xi_1 \d \xi_2 \d \xi_3 \right) .  \label{ed:def-K-G-alpha}
\end{align}

\begin{theo}[Self-adjointness]\label{thm:self-adj}
Assume that Hypothesis \ref{hyp:1} holds. Let $g_0>0$ be such that
\begin{equation}\label{eq:g0-cond}
  g_0{}^2\left(\sum_{\alpha=1,2} K_1(G^{(\alpha)})^2\right) (\frac{1}{\me^2}+1) <1\, .
\end{equation}
Then for any real $g$ such that $|g|\leq g_0$, the operator $H = H_0 + gH_I$ is self-adjoint with domain $\D(H_0)$. Moreover, any core for $H_0$ is a core for $H$.
\end{theo}
\begin{rema}
\begin{itemize}
\item[1)] Combining \eqref{eq:spec_H0} and standard perturbation theory of isolated eigenvalues (see e.g. \cite{Kato}), we deduce that, for $|g| \ll \me$, $\inf \mathrm{spec}( H )$ is a non-degenerate eigenvalue of $H$. In other words, $H$ admits a unique ground state.
\item[2)] Let $Q$ be the total charge operator
\begin{equation*}
  Q= \Ne - \Np ,
\end{equation*}
where $\Ne$ and $\Np$ are respectively the operator number of electrons and the operator number of positrons given by \eqref{eq:def-NpNe}.

The total Hamiltonian $H$ commutes with $Q$, and $\mathcal{H}$ is decomposed with respect to the spectrum of the total charge operator as
\begin{equation*}
  \mathcal{H} \simeq \oplus_{z\in\Z}\mathcal{H}_{z}.
\end{equation*}
Each $\mathcal{H}_{z}$ reduces $H$ and by mimicking the proof given in \cite{Takaesu2009} one proves that the ground state of $H$ belongs to
$\mathcal{H}_{0}$.
\end{itemize}
\end{rema}

Theorem \ref{thm:self-adj} follows from the Kato-Rellich Theorem together with standard estimates of creation and annihilation operators in Fock space, showing that the interaction Hamiltonian $H_I$ is relatively bounded with respect to $H_0$. For the convenience of the reader, a sketch of the proof of Theorem \ref{thm:self-adj} is recalled in Subsection \ref{subsection:self-adj}.

For a self-adjoint operator $A$, we denote by $\mathrm{spec}_{ \mathrm{ess} }( A )$ the essential spectrum of $A$.

\begin{theo}[Localization of the essential spectrum]\label{thm:ess}
Assume that Hypothesis \ref{hyp:1} holds and let $g_0$ be as in \eqref{eq:g0-cond}. Then, for all $|g| \le g_0$,
\begin{equation*}
\mathrm{spec}_{ \mathrm{ess} } ( H ) = [ \inf \mathrm{spec} ( H ) + \me , \infty ).
\end{equation*}
\end{theo}
Theorem \ref{thm:ess} is proven in Subsection \ref{subsection:ess}.  Our proof is based on a method due to Derezi{\'n}ski and G{\'e}rard \cite{DG} that we adapt to our context.

To establish our next theorems, we need to strengthen the conditions on the kernels $G^{(\alpha)}$.

\begin{hypo}\label{hyp:2}
For $\alpha=1,2$, the kernels $G^{(\alpha)} \in L^2 ( \Sigma \times \Sigma \times \Sigma_3 )$ satisfy
\begin{itemize}
\item[(i)] There exists a compact set $K \subset \mathbb{R}_+ \times \mathbb{R}_+ \times \mathbb{R}^3$ such that $G^{(\alpha)} ( p_1 , \gamma_1 , p_2 , \gamma_2 , k , \lambda ) = 0$ if $( p_1 , p_2 , k ) \notin K$.
\item[(ii)] There exists $\varepsilon \ge 0$ such that
\begin{align*}
\sum_{ \gamma_1 , \gamma_2 , \lambda } \int ( 1 + x_1^2 + x_2^2 )^{1 + \varepsilon } \left | \hat G^{(\alpha)}( x_1 , \gamma_1 , x_2 , \gamma_2 , k , \lambda )\right|^2 \d x_1 \d x_2 \d k < \infty ,
\end{align*}
where $\hat G^{(\alpha)}$ denote the Fourier transform of $G^{(\alpha)}$ with respect to the variables $(p_1,p_2)$, and $x_j$ is the variable dual to $p_j$.
\item[(iii)] If $\gamma_{1j} = 1$ or $\gamma_{2j} = 1$, where for $l=1,2$, $ \gamma_{lj} = |Ê\kappa_{j_l} |$ (with $\gamma_{l} = ( j_l , m_{j_l} , \kappa_{j_l} )$), and if $p_1 = 0$ or $p_2 = 0$, then $G^{(\alpha)} ( p_1 , \gamma_1 , p_2 , \gamma_2 , k , \lambda ) = 0$.
\end{itemize}
\end{hypo}
\begin{rema}
\begin{itemize}
\item[1)] The assumption that $G^{(\alpha)}$ is compactly supported in the variables $(p_1,p_2,k)$ is an ``ultraviolet'' constraint that is made for convenience. It could be replaced by the weaker assumption that $G^{(\alpha)}$ decays sufficiently fast at infinity.
\item[2)] Hypothesis \ref{hyp:2} $(ii)$ comes from the fact that the coupling functions $G^{(\alpha)}$ must satisfy some ``minimal'' regularity for our method to be applied. In fact, Hypothesis $(ii)$ could be slightly improved with a refined choice of interpolation spaces in our proof (see Section \ref{sec:LAP} for more details). In Hypothesis \ref{hyp:2} $(iii)$, we need in addition an ``infrared'' regularization. We remark in particular that Hypotheses $(ii)$ and $(iii)$ imply that, for $0 \le \varepsilon < 1/2$,
\begin{align*}
\big | G^{(\alpha)} ( p_1 , \gamma_1 , p_2 , \gamma_2 , k , \lambda ) \big | \lesssim | p_l |^{\frac12 + \varepsilon}, \quad l = 1,2.
\end{align*}
We emphasize, however, that this infrared assumption is required only in the case $\gamma_{lj} = 1$, that is, for $j= 1/2$. For all other $j \in \mathbb{N} + \frac12$, we do \emph{not} need to impose any infrared regularization on the generalized eigenstates $\psi_{\pm}( p , \gamma )$; They are already regular enough.
\item[3)] One verifies that Hypotheses \ref{hyp:2}(i) and \ref{hyp:2}(ii) imply Hypothesis \ref{hyp:1}.
\end{itemize}
\end{rema}
\begin{theo}[Localization of the spectrum]\label{thm:loc-spec}
Assume that Hypothesis \ref{hyp:2} holds. There exists $g_1 > 0$ such that, for all $|g| \le g_1$,
\begin{equation*}
\mathrm{spec} ( H ) = \{ \inf \mathrm{spec} ( H ) \} \cup  [ \inf \mathrm{spec} ( H ) + \me , \infty ).
\end{equation*}
In particular, $H$ has no eigenvalue below its essential spectrum except for the ground state energy, $\inf \mathrm{spec} ( H )$, which is a simple eigenvalue.
\end{theo}
\begin{theo}[Absolutely continuous spectrum]\label{thm:abs-cont}
Assume that Hypothesis \ref{hyp:2} holds with $\varepsilon > 0$ in Hypothesis \ref{hyp:2}(ii). For all $\delta > 0$, there exists $g_\delta > 0$ such that, for all $| g | \le g_\delta$, the spectrum of $H$ in the interval
\begin{equation*}
[ \inf \mathrm{spec}( H ) + \me ,\, \inf \mathrm{spec}( H ) + \mz - \delta ]
\end{equation*}
is purely absolutely continuous.
\end{theo}
\begin{rema}
\begin{itemize}
\item[1)] In Theorem \ref{thm:LAP} below, we prove a stronger result than Theorem \ref{thm:abs-cont}, which is of independent interest, namely we show that a \emph{limiting absorption principle} holds for $H$ in the interval $[ \inf \mathrm{spec}( H ) + \me ,\, \inf \mathrm{spec}( H ) + \mz - \delta ]$. Another consequence of the limiting absorption principle of Theorem \ref{thm:LAP} is the local decay property \eqref{eq:local-decay}.
\item[2)] If we make the further assumption that the kernels $G^{(\alpha)}$ are sufficiently regular with respect to the $Z^0$ variable $k$, similarly to what is assumed in Hypothesis \ref{hyp:2}(ii) for the variables $p_1$, $p_2$, we can presumably extend the result of Theorem \ref{thm:abs-cont} to the interval $[ \inf \mathrm{spec} ( H ) + \me , M )$ for any $M > \inf \mathrm{spec} ( H ) + \me$.
\end{itemize}
\end{rema}
Theorems \ref{thm:loc-spec} and \ref{thm:abs-cont} are proven in Section \ref{sec:LAP}. Our proofs rely on Mourre's Theory with a non-self adjoint conjugate operator. Such extensions of the usual conjugate operator theory \cite{Mo,ABG} have been considered in \cite{HuSp}, \cite{Sk}, and later extended in \cite{GGM1,GGM2}.

We use in this paper a conjugate operator, $A$, similar to the ones of \cite{HuSp} and \cite{GGM1,GGM2}, and prove regularity of the total Hamiltonian with respect to this conjugate operator. Combined with a Mourre estimate, this regularity property allows us to deduce a virial theorem and a limiting absorption principle, from which we obtain Theorems \ref{thm:loc-spec} and \ref{thm:abs-cont}.

Our main achievement consists in proving that the physical interaction Hamiltonian $H_I$ is regular enough for the Mourre theory to be applied, except for the terms associated to the ``first'' generalized eigenstates ($j=1/2$). For the latter, we need to impose a non-physical infrared condition. To establish the regularity of $H_I$ with respect to $A$, we use in particular real interpolation theory, together with a version of the Mourre theory requiring only low regularity of the Hamiltonian with respect to the conjugate operator.

In Appendix \ref{appendixA}, we give the estimates on the generalized eigenfunctions of the Dirac operator that are used in this paper. In Appendix \ref{appendix:Mourre}, we recall the abstract results from Mourre's theory that we need. Finally, for the convenience of the reader, standard definitions and properties of creation and annihilation operators in Fock space are recalled in Appendix \ref{appendixC}.

\section{Self-adjointness and localization of the essential spectrum}
In this section we prove Theorems \ref{thm:self-adj} and \ref{thm:ess}.

\subsection{Self-adjointness}\label{subsection:self-adj}
We sketch the standard proof of Theorem \ref{thm:self-adj} relying on the Kato-Rellich Theorem.
\begin{proof}[Proof of Theorem \ref{thm:self-adj}]
We use the $N_\tau$ estimates of \cite{GlimmJaffe} and follow the proof of \cite[Theorem~2.6]{BDG2} (see also \cite{BDG}). For
\begin{equation}\label{eq:def-kig}
 K_i(G)^2 := \sum_{\alpha=1,2} K_i(G^{(\alpha)})^2,\quad i=1,2\, ,
\end{equation}
and
\begin{equation}\nonumber
\begin{split}
 C_{1,\beta} & := (\frac{1}{\me^2} +1+2\beta)^\frac12,\quad
 C_{2,\beta\eta} := ( \frac{\eta}{\me^2} (1+2\beta))^\frac12, \\
 B_{1,\beta} & := (1+\frac{1}{2\beta})^\frac12,\quad
 B_{2,\beta\eta} :=(\eta(1+\frac{1}{2\beta}) + \frac{1}{4\eta})^\frac12\, ,
\end{split}
\end{equation}
we obtain, for any $\psi\in\D(H)$,
\begin{equation}\label{eq:rel-bound-1}
\begin{split}
  \| H_I\psi\| \leq
  \left(K_1(G)\, C_{1,\beta}+ K_2(G) C_{2,\beta}\right)\|H_0\psi\|
  + \left(K_1(G) B_{1,\beta} + K_2(G)B_{2,\beta\eta}\right) \|\psi\| \, .
\end{split}
\end{equation}
Therefore, with \eqref{eq:g0-cond} and for $\beta$ and $\eta$ small enough, using the Kato-Rellich Theorem
concludes the proof.
\end{proof}
If we note that $K_2(G)\geq K_1(G)$, and set
$$
 K(G) := K_2(G)\, , \quad C_{\beta\eta}:=C_{1,\beta} + C_{2,\beta\eta}\, ,\quad
 B_{\beta\eta} := B_{1,\beta} + B_{2,\beta\eta}\, ,
$$
we obtain from \eqref{eq:rel-bound-1} the following relative bound:
\begin{corr}\label{cor-rel-bound}
For any $\psi\in\D(H)$,
$$
 \|H_I\psi\| \leq K(G) \left( C_{\beta\eta}\|H_0\psi\| + B_{\beta\eta}\|\psi\|\right)\, .
$$
\end{corr}
In the sequel, for the sake of simplicity, we shall use this relative bound instead of
the stronger result \eqref{eq:rel-bound-1}.

\subsection{Localization of the essential spectrum}\label{subsection:ess}
In this subsection, we prove Theorem \ref{thm:ess}. We use the Derezi{\'n}ski-G{\'e}rard partition of unity \cite{DG} in a version that accommodates the Fermi-Dirac statistics and the CAR (such a partition of unity was used previously in \cite{ammari}). Let
\begin{equation*}
U_a : \mathfrak{F}_a ( \mathfrak{H}_c \oplus \mathfrak{H}_c ) \to \mathfrak{F}_a ( \mathfrak{H}_c ) \otimes \mathfrak{F}_a( \mathfrak{h}_c ) = \mathfrak{F}_a \otimes \mathfrak{F}_a ,
\end{equation*}
be defined by
\begin{align*}
& U_a \Omega_a = \Omega_a \otimes \Omega_a \\
& U_a b^*( \varphi_1 \oplus \varphi_2 ) = ( b^*( \varphi_1 ) \otimes \1 + ( -1 )^{ N } \otimes b^*( \varphi_2 ) ) U_a ,
\end{align*}
where $(-1)^{N}$ denotes the bounded operator on $\mathfrak{F}_a$ defined by its restriction to $\otimes_a^r \mathfrak{h}_c$ as $(-1)^{N} u = (-1)^r u$ for any $u \in \otimes_a^r \mathfrak{h}_c$. Clearly, using the anti-commutation relations, $U_a$ extends by linearity to a unitary map on $\mathfrak{F}_a ( \mathfrak{H}_c \oplus \mathfrak{H}_c )$. Let $j_0 \in \mathrm{C}^\infty( [ 0 , \infty ) ; [ 0 , 1 ] )$ be such that $j_0 \equiv 1$ on $[ 0 , 1/2 ]$ and $j_0 \equiv 0$ on $[ 1 , \infty )$, and let $j_\infty$ be defined by the relation $j_0^2 + j_\infty^2 \equiv 1$. Let $y := i \nabla_p$ account for the position variable of the fermions. Given $R > 0$, we introduce the bounded operators $j_0^R := j_0 ( |y| / R )$ and $j_\infty^R := j_\infty( |y| / R )$ on $\mathfrak{F}_a ( \mathfrak{H}_c )$. Let
\begin{align*}
j_a^R : \mathfrak{H}_c &\to \mathfrak{H}_c \oplus \mathfrak{H}_c \\
\varphi & \mapsto ( j_0^R \varphi , j_\infty^R \varphi ).
\end{align*}
Lifting the operator $j_a^R$ to the Fock space $\mathfrak{F}_a ( \mathfrak{H}_c )$ allows one to define a map $\Gamma( j_a^R ) : \mathfrak{F}_a ( \mathfrak{H}_c ) \to \mathfrak{F}_a ( \mathfrak{H}_c \oplus \mathfrak{H}_c  )$. The Derezi{\'n}ski-G{\'e}rard partition of unity is defined by
\begin{equation*}
\check{\Gamma}_a(j_a^R) : \mathfrak{F}_a \to \mathfrak{F}_a \otimes \mathfrak{F}_a , \qquad \check{\Gamma}_a(j_a^R) = U_a \Gamma( j_a^R ).
\end{equation*}
Using the relation $j_0^2 + j_\infty^2 \equiv 1$, one easily verifies that $\check{\Gamma}_a(j_a^R)$ is isometric.

We construct a similar partition of unity, $\check{\Gamma}_s( j_s^R )$, acting on the bosonic Fock space $\gF_{\z0} = \mathfrak{F}_s ( L^2(\Sigma_3) )$. It is defined by
\begin{equation*}
\check{\Gamma}_s(j_s^R) : \gF_{\z0} \to \gF_{\z0} \otimes \gF_{\z0} , \qquad \check{\Gamma}_s(j_s^R) = U_s \Gamma( j_s^R ) ,
\end{equation*}
where
\begin{equation*}
U_s : \mathfrak{F}_s ( L^2(\Sigma_3) \oplus L^2(\Sigma_3) ) \to \gF_{\z0} \otimes \gF_{\z0} ,
\end{equation*}
is the unitary operator defined by
\begin{align*}
& U_s \Omega_s = \Omega_s \otimes \Omega_s \\
& U_s a^*( \varphi_1 \oplus \varphi_2 ) = ( a^*( \varphi_1 ) \otimes \1 + \1 \otimes a^*( \varphi_2 ) ) U_s ,
\end{align*}
and $j_s^R$ is the bounded operator defined by
\begin{align*}
j_s^R : L^2(\Sigma_3) &\to L^2(\Sigma_3) \oplus L^2(\Sigma_3) \\
\varphi & \mapsto ( j_0^R \varphi , j_\infty^R \varphi ).
\end{align*}
Here we have used similar notations as above, namely $j_0^R := j_0 ( |y_3| / R )$ and $j_\infty^R := j_\infty( |y_3| / R )$, where $y_3 := i \nabla_k$ accounts for the position variable of the bosons.

Let $N$ denote the number operator, acting either on $\mathfrak{F}_a$ or on $ \gF_{\z0}$. To shorten notations, we define the operators
\begin{equation*}
N_0 := N \otimes \1 , \qquad N_\infty := \1 \otimes N ,
\end{equation*}
acting on $\mathfrak{F}_a \otimes \mathfrak{F}_a $ and on $ \gF_{\z0} \otimes  \gF_{\z0}$.

We recall the following properties that can be easily proven using the definitions of the operators involved (see \cite{ammari,DG}).
\begin{lem}\label{lm:checkGamma}
With the previous notations, we have the following properties.
\begin{itemize}
\item[(i)] Let $\varphi_1,\dots,\varphi_n \in \mathfrak{H}_c$. Then
\begin{align*}
\check{\Gamma}_a( j_a^R ) \prod_{i=1}^n b^*( \varphi_i ) \Omega_a = \prod_{i=1}^n \big ( b^*( j_0^R \varphi_i ) \otimes \1 + (-1)^N \otimes b^*( j_\infty^R \varphi_i ) \big ) \Omega_a \otimes \Omega_a.
\end{align*}
Let $\varphi_1,\dots,\varphi_n \in  L^2(\Sigma_3)$. Then
\begin{align*}
\check{\Gamma}_s( j_s^R ) \prod_{i=1}^n a^*( \varphi_i ) \Omega_s = \prod_{i=1}^n \big ( a^*( j_0^R \varphi_i ) \otimes \1 + \1 \otimes a^*( j_\infty^R \varphi_i ) \big ) \Omega_s \otimes \Omega_s.
\end{align*}
\item[(ii)] Let $\omega$ be an operator on $\mathfrak{H}_c$ such that the commutators $[ \omega , j^R_\# ]$, defined as quadratic forms on $\mathfrak{D}( \omega )$, extend to bounded operators on $\mathfrak{H}_c$, where $j_\#$ stands for $j_0$ and $j_\infty$. Then
\begin{align*}
& \big \| ( N_0 + N_\infty )^{-\frac12} \big ( ( \d \Gamma( \omega ) \otimes \1 + \1 \otimes \d \Gamma ( \omega ) ) \check{\Gamma}_a( j_a^R ) - \check{\Gamma}_a( j_a^R ) \d\Gamma ( \omega ) \big ) N^{-\frac12} P_{\Omega_a}^\perp  \big \| \\
& \le \big \| \check{ \mathrm{ad} }_\omega ( j_a^R ) \big \| ,
\end{align*}
where $P_{\Omega_a}$ denotes the orthogonal projection onto the vacuum sector in $\mathfrak{F}_a$, and $\check{ \mathrm{ad} }_\omega ( j_a^R ) := ( [ \omega , j^R_0 ] , [ \omega , j^R_\infty ] )$.

The same estimate holds if  $\mathfrak{F}_a$, $\mathfrak{H}_c$, $j_a^R$, $\check{\Gamma}_a$ and $\Omega_a$ are replaced respectively by $\gF_{\z0}$, $L^2( \Sigma_3)$, $j_s^R$, $\check{\Gamma}_s$ and $\Omega_s$.
\end{itemize}
\end{lem}

Recall that the total Hilbert space can be written as $\mathcal{H} = \mathfrak{F}_a \otimes \mathfrak{F}_a \otimes \mathfrak{F}_{Z^0}$. As in \cite{ammari,DG}, it is convenient to introduce an ``extended'' Hamiltonian, $H^{\mathrm{ext}}$, acting on the ``extended'' Hilbert space
\begin{equation*}
\mathcal{H}^{\mathrm{ext}} := \bigotimes_{i=1}^4 \mathfrak{F}_a \otimes  \bigotimes_{j=1}^2 \mathfrak{F}_{Z^0} .
\end{equation*}
In our setting, the extended Hamiltonian is given by the expression
\begin{equation*}
H^{\mathrm{ext}} := H_0^{\mathrm{ext}} + g H_I^{ \mathrm{ext} } ,
\end{equation*}
where
\begin{align*}
H_0^{\mathrm{ext}} :=& \d\Gamma( H_D ) \otimes \1_{ \otimes^2 \mathfrak{F}_a } \otimes \1_{ \otimes^2 \mathfrak{F}_{Z^0} }  + \1_{ \otimes^2 \mathfrak{F}_a } \otimes \d\Gamma( H_D ) \otimes \1_{ \otimes^2 \mathfrak{F}_{Z^0} }  \\
&+ \1_{ \otimes^4 \mathfrak{F}_a } \otimes \d \Gamma ( H_{Z^0} ) \otimes \1_{ \mathfrak{F}_{Z^0} } + \1_{ \otimes^4 \mathfrak{F}_a } \otimes  \1_{ \mathfrak{F}_{Z^0} } \otimes \d \Gamma ( H_{Z^0} ),
\end{align*}
and $H_I^{ \mathrm{ext} }$ is given by \eqref{eq:interaction}--\eqref{eq:interaction-4}, except that the creation and annihilation operators for the electrons, $b_+^\# = b^\# \otimes \1 \otimes \1$, acting on $\mathcal{H} = \mathfrak{F}_a \otimes \mathfrak{F}_a \otimes \mathfrak{F}_{Z^0}$, are replaced by
\begin{equation*}
b_+^{\#,0} := b^\# \otimes \1_{ \otimes^3 \mathfrak{F}_a } \otimes \1_{ \otimes^2 \mathfrak{F}_{Z^0} }
\end{equation*}
(acting on $\mathcal{H}^{ \mathrm{ext} }$), likewise, the creation and annihilation operators for the positrons, $b^\#_- = (-1)^{\Ne} \otimes b^\# \otimes \1$, are replaced by
\begin{equation*}
b^{\#, 0 }_- := (-1)^{\Nez} \otimes (-1)^{\Nei} \otimes b^\# \otimes \1_{ \mathfrak{F}_a } \otimes \1_{ \otimes^2 \mathfrak{F}_{Z^0} } ,
\end{equation*}
and the creation and annihilation operators for the bosons, $a^\#$, are replaced by
\begin{equation*}
a^{\#, 0 } := \1_{ \otimes^4 \mathfrak{F}_a } \otimes a^\# \otimes \1_{ \mathfrak{F}_{Z^0} } .
\end{equation*}
Here we have set
\begin{align*}
\Nez := \left(N \otimes \1_{ \mathfrak{F}_a }\right) \otimes
\1_{ \otimes^2 \mathfrak{F}_a } \otimes \1_{ \otimes^2 \mathfrak{F}_{Z^0} } , \qquad \Nei := \left(\1_{ \mathfrak{F}_a } \otimes N \right) \otimes \1_{ \otimes^2 \mathfrak{F}_a } \otimes \1_{ \otimes^2 \mathfrak{F}_{Z^0} } ,
\end{align*}
on $\mathcal{H}^{ \mathrm{ext} }$. We define similarly the number operators
\begin{align*}
\Npz := \1_{ \otimes^2 \mathfrak{F}_a } \otimes \left(N \otimes \1_{ \mathfrak{F}_a }\right) \otimes \1_{ \otimes^2 \mathfrak{F}_{Z^0} }, \qquad
\Npi := \1_{ \otimes^2 \mathfrak{F}_a } \otimes \left(\1_{\mathfrak{F}_a }
\otimes N \right) \otimes \1_{ \otimes^2 \mathfrak{F}_{Z^0} } ,
\end{align*}
and
\begin{align*}
N_{\z0,0} := \1_{ \otimes^4 \mathfrak{F}_a } \otimes \left( N \otimes \1_{ \mathfrak{F}_{Z^0} }\right), \qquad
N_{ \z0 , \infty} := \1_{ \otimes^4 \mathfrak{F}_a } \otimes \left(\1_{ \mathfrak{F}_{Z^0} } \otimes N \right) ,
\end{align*}
and the creation and annihilation operators
\begin{equation*}
b_+^{\#,\infty} := \1_{ \mathfrak{F}_a } \otimes b^\# \otimes \1_{ \otimes^2 \mathfrak{F}_a } \otimes \1_{ \otimes^2 \mathfrak{F}_{Z^0} } ,
\end{equation*}
\begin{equation*}
b^{\#, \infty }_- := (-1)^{\Nez} \otimes (-1)^{\Nei} \otimes \1_{ \mathfrak{F}_a } \otimes b^\# \otimes \1_{ \otimes^2 \mathfrak{F}_{Z^0} } ,
\end{equation*}
and
\begin{equation*}
a^{\#, \infty } := \1_{ \otimes^4 \mathfrak{F}_a } \otimes \1_{ \mathfrak{F}_{Z^0} } \otimes a^\#  .
\end{equation*}

Now, we introduce an isometric map, $\check{\Gamma}_R : \mathcal{H} \to \mathcal{H}^{ \mathrm{ext} }$, by setting
\begin{equation*}
\check{\Gamma}_R := \check{\Gamma}_a( j_a^R ) \otimes \check{\Gamma}_a( j_a^R ) \otimes \check{\Gamma}_s( j_s^R ).
\end{equation*}
Theorem \ref{thm:ess} will be a consequence of the following lemma.
\begin{lem}\label{lm:Hext}
Assume that Hypothesis 1 holds and let $g_0$ be as in \eqref{eq:g0-cond}. Let $\chi \in \mathrm{C}_0^\infty( \mathbb{R} )$. Then, for all $|g| \le g_0$,
\begin{equation*}
\big \| \check{\Gamma}_R \chi( H ) - \chi( H^{\mathrm{ext}} ) \check{\Gamma}_R \big \| \to 0 , \quad \text{as } R \to \infty.
\end{equation*}
\end{lem}
\begin{proof}
Using the Helffer-Sj{\"o}strand functional calculus, we represent $\chi(H)$ as the integral
\begin{equation*}
\chi( H ) = \frac{ 1 }{ \pi } \int \frac{ \partial \tilde \chi }{ \partial \bar z }  ( z ) ( H - z )^{-1} \d \, \mathrm{Re} z \, \d \, \mathrm{Im} z ,
\end{equation*}
where $\tilde \chi \in \mathrm{C}_0^\infty( \mathbb{C} )$ denotes an almost analytic extension of $\chi$ satisfying $\tilde \chi |_{ \mathbb{R} } = \chi$ and $| \partial_{ \bar z } \tilde \chi( z ) | \le C_n | \mathrm{Im} \, z |^n$ for any $n \in \mathbb{N}$. The same representation holds for $\chi( H^{ \mathrm{ext} } )$, from which we deduce that
\begin{align*}
& \check{\Gamma}_R \chi( H ) - \chi( H^{\mathrm{ext}} ) \check{\Gamma}_R \\
&= \frac{ 1 }{ \pi } \int \frac{ \partial \tilde \chi }{ \partial \bar z }  ( z ) ( H^{ \mathrm{ext} } - z )^{-1} ( H^{ \mathrm{ext} } \check{\Gamma}_R - \check{\Gamma}_R H ) ( H - z )^{-1} \d \, \mathrm{Re} z \, \d \, \mathrm{Im} z.
\end{align*}
By Lemma \ref{lm:checkGamma}$(ii)$, together with
\begin{equation*}
\big \| N_\#^{\frac12} ( H - z )^{-1} \big \| \le C | \mathrm{Im} \, z |^{-1} , \qquad \big \| ( H^{ \mathrm{ext} } - z )^{-1} ( N_{\#,0} + N_{\#,\infty} )^{\frac12} \big \| \le C | \mathrm{Im} \, z |^{-1} ,
\end{equation*}
where $N_{\#}$ stands for $\Ne$, $\Np$ or $N_{ \z0 }$  (and likewise for $N_{ \# , 0 }$ and $N_{ \# , \infty }$), we obtain
\begin{align}
& \big \| ( H^{ \mathrm{ext} } - z )^{-1} ( H_0^{ \mathrm{ext} } \check{\Gamma}_R - \check{\Gamma}_R H_0 ) ( H - z )^{-1} \big \| \notag \\
& \le C \big ( \big \| \check{ \mathrm{ad} }_\omega ( j_a^R ) \big \| + \big \| \check{ \mathrm{ad} }_{\omega_3} ( j_s^R ) \big \| \big ) | \mathrm{Im} \, z |^{-2} .\label{eq:a1}
\end{align}
Here, $\omega$ is given by \eqref{eq:3} and $\omega_3$ is given by \eqref{eq:omega3}. Using e.g. pseudo-differential calculus, one easily verifies that $\big \| \check{ \mathrm{ad} }_\omega ( j_a^R ) \big \| = \mathcal{O}( R^{-1} )$ and $\big \| \check{ \mathrm{ad} }_{\omega_3} ( j_s^R ) \big \| = \mathcal{O}( R^{-1} )$, as $R \to \infty$. Hence, \eqref{eq:a1} combined with the properties of the almost analytic extension $\tilde \chi$ show that
\begin{align*}
\Big \| \int \frac{ \partial \tilde \chi }{ \partial \bar z }  ( z ) ( H^{ \mathrm{ext} } - z )^{-1} ( H_0^{ \mathrm{ext} } \check{\Gamma}_R - \check{\Gamma}_R H_0 ) ( H - z )^{-1} \d \, \mathrm{Re} z \, \d \, \mathrm{Im} z \Big \| = \mathcal{O}( R^{-1} ).
\end{align*}

It remains to estimate
\begin{align*}
\int \frac{ \partial \tilde \chi }{ \partial \bar z }  ( z ) ( H^{ \mathrm{ext} } - z )^{-1} ( H_I^{ \mathrm{ext} } \check{\Gamma}_R - \check{\Gamma}_R H_I ) ( H - z )^{-1} \d \, \mathrm{Re} z \, \d \, \mathrm{Im} z .
\end{align*}
The different interaction terms appearing in the definition \eqref{eq:interaction} of $H_I$ are treated in the same way. Consider for instance the interaction Hamiltonian $H_I^{(1)}$ given by \eqref{eq:interaction-1}, written under the form given in \eqref{eq:interaction-1b},
\begin{equation*}
 H_I^{(1)} = \int F^{(1)}(\xi_1,\, \xi_2,\, \xi_3) b_+^*(\xi_1) b_-^*(\xi_2) a(\xi_3)\, \d\xi_1\d\xi_2\d\xi_3\, ,
\end{equation*}
with $F^{(1)} \in L^2( \d\xi_1\d\xi_2\d\xi_3 )$. We let $H_I^{(1), \mathrm{ext} }$ be defined by the same expression, except that the creation and annihilation operators $b_+^*$,  $b_-^*$, $a$ are replaced by $b_+^{* , 0}$,  $b_-^{* , 0}$, $a^{0}$ defined above.  Using Lemma \ref{lm:checkGamma}$(i)$, a straightforward computation gives
\begin{align*}
& H_I^{ (1) , \mathrm{ext} } \check{\Gamma}_R - \check{\Gamma}_R H_I^{(1)} \\
& = \int  j_1( | i \nabla_{p_1} | , | i \nabla_{p_2} | , | i \nabla_k | ) F^{(1)}(\xi_1,\, \xi_2,\, \xi_3) b_+^{*, 0 }(\xi_1) b_-^{*, 0 }(\xi_2) a^{0}(\xi_3) \check{\Gamma}_R \, \d\xi_1\d\xi_2\d\xi_3 \, , \\
& + \sum_{l > 1} \int  j_l( | i \nabla_{p_1} | , | i \nabla_{p_2} | , | i \nabla_k | ) F^{(1)}(\xi_1,\, \xi_2,\, \xi_3) b_+^{*, \sharp }(\xi_1) b_-^{*, \sharp }(\xi_2) a^{ \sharp }(\xi_3) \check{\Gamma}_R \, \d\xi_1\d\xi_2\d\xi_3 \, ,
\end{align*}
where we have set $j_1( |y_1| , |y_2| ,| y_3| ) = 1 - j_0( |y_1| / R ) j_0( |y_2| / R ) j_0( |y_3| / R )$ and, for $l \neq 1$, $j_l( |y_1| , |y_2| , |y_3| )$ is of the form $j_l( |y_1| , |y_2| , |y_3| ) = j_{\#1}( |y_1| / R ) j_{\#2}( |y_2| / R )j_{\#3}( |y_3| / R )$ with $j_{\# i} = j_0$ or $j_{\# i} =j_\infty$, and at least one of the $j_{\#i}$'s is equal to $j_\infty$. Moreover, $b_+^{*,\sharp}$ stands for $b_+^{*,0}$ or $b_+^{*,\infty}$, and likewise for $b_-^{*,\sharp}$ and $a^{\sharp}$.

It follows from the $N_\tau$ estimates (see \cite{GlimmJaffe}) that
\begin{align*}
& \big \| ( H^{ \mathrm{ext} } - z )^{-1} ( H_I^{ (1) , \mathrm{ext} } \check{\Gamma}_R - \check{\Gamma}_R H_{I,(1)} ) ( H - z )^{-1} \big \| \\
& \le C | \mathrm{Im} \, z |^{-2} \sum_{l} \big \|  j_l( | i \nabla_{p_1} | , | i \nabla_{p_2} | , | i \nabla_k |) F^{(1)} \big \|.
\end{align*}
Therefore, using the fact that
\begin{equation*}
\big \|  j_l( | i \nabla_{p_1} | , | i \nabla_{p_2} | , | i \nabla_k |) F^{(1)} \big \| \to 0 ,
\end{equation*}
as $R \to \infty$ and the properties of $\tilde \chi$, we deduce that
\begin{align*}
\Big \| \int \frac{ \partial \tilde \chi }{ \partial \bar z }  ( z ) ( H^{ \mathrm{ext} } - z )^{-1} ( H_I^{ (1) , \mathrm{ext} } \check{\Gamma}_R - \check{\Gamma}_R H^{(1)}_I ) ( H - z )^{-1} \d \, \mathrm{Re} z \, \d \, \mathrm{Im} z \Big \| \to 0 ,
\end{align*}
as $R \to \infty$. Since the other interaction terms in \eqref{eq:interaction} are treated in the same way, this concludes the proof.
\end{proof}
We are now ready to prove Theorem \ref{thm:ess}.
\begin{proof}[Proof of Theorem \ref{thm:ess}] We prove that
\begin{equation}\label{eq:firstincl}
\mathrm{spec}_{ \mathrm{ess} } ( H ) \subset [ \inf \mathrm{spec} ( H ) + \me , \infty ).
\end{equation}
Let $\chi \in \mathrm{C}_0^\infty( ( - \infty , \inf \mathrm{spec} ( H ) + \me ))$. Since $\check{\Gamma}_R$ is isometric, we can write
\begin{align}\label{eq:a2}
\chi ( H ) = \check{\Gamma}_R^* \check{\Gamma}_R \chi ( H ) = \check{\Gamma}_R^* \chi( H^{ \mathrm{ext} } ) \check{\Gamma}_R + o_R(1) ,
\end{align}
where $o_R(1)$ stands for a bounded operator vanishing as $R \to \infty$. The last equality above follows from Lemma \ref{lm:Hext}. Observing that $N_{\mathrm{tot} , \infty }  := \Nei + \Npi + N_{ \z0 , \infty }$ commutes with $H^{ \mathrm{ext} }$ and that
\begin{equation*}
H^{ \mathrm{ext} } \1_{ [ 1 , \infty ) }( N_{\mathrm{tot} , \infty } ) \ge ( \inf \mathrm{spec} ( H ) + \me ) \1_{ [ 1 , \infty ) }( N_{\mathrm{tot} , \infty } ),
\end{equation*}
we deduce that
\begin{equation*}
\chi \big ( H^{ \mathrm{ext} } \big ) = \1_{ \{ 0 \} }( N_{\mathrm{tot} , \infty } ) \chi \big ( H^{ \mathrm{ext} } \big ).
\end{equation*}
Hence \eqref{eq:a2} yields
\begin{align}
\chi ( H ) & = \check{\Gamma}_R^* \1_{ \{ 0 \} }( N_{\mathrm{tot} , \infty } ) \chi \big ( H^{ \mathrm{ext} } \big ) \check{\Gamma}_R + o_R(1) \notag \\
& = \check{\Gamma}_R^* \1_{ \{ 0 \} }( N_{\mathrm{tot} , \infty } ) \check{\Gamma}_R \chi( H ) + o_R(1) , \label{eq:a3}
\end{align}
where we used again Lemma \ref{lm:Hext} in the last equality. Inspecting the definition of the operator $\check{\Gamma}_R$, it is easy to see that
\begin{equation*}
\check{\Gamma}_R^* \1_{ \{ 0 \} }( N_{\mathrm{tot} , \infty } ) \check{\Gamma}_R = \Gamma \big ( ( j_0^R )^2 \big ) \otimes \Gamma \big ( ( j_0^R )^2 \big ) \otimes \Gamma \big ( ( j_0^R )^2 \big ) .
\end{equation*}
Since
\begin{equation*}
\Gamma \big ( ( j_0^R )^2 \big ) \otimes \Gamma \big ( ( j_0^R )^2 \big ) \otimes \Gamma \big ( ( j_0^R )^2 \big ) ( H_0 + i )^{-1}
\end{equation*}
is compact, and since $(H_0 + i ) \chi( H )$ is bounded, we conclude that
\begin{equation*}
\check{\Gamma}_R^* \1_{ \{ 0 \} }( N_{\mathrm{tot} , \infty } ) \check{\Gamma}_R \chi( H )
\end{equation*}
is compact. Therefore, by \eqref{eq:a3}, the operator $\chi( H )$ is also compact, which proves \eqref{eq:firstincl}.

To prove the converse inclusion, it suffices to construct, for any $\lambda \in ( \inf \mathrm{spec} ( H ) + \me , \infty )$, a Weyl sequence associated to $\lambda$. This can be done in the same way as in \cite[Theorem 4.1]{DG} or \cite[Theorem 4.3]{ammari}. We do not give the details.
\end{proof}

\section{Proofs of Theorems \ref{thm:loc-spec} and \ref{thm:abs-cont}}\label{sec:LAP}

In this section, we prove Theorems \ref{thm:loc-spec} and \ref{thm:abs-cont} by applying a suitable version of Mourre's theory. We begin with defining the conjugate operator $A$ that we consider in Subsection \ref{subsection:conjugate}; We show that the semi-group generated by $A$ preserves the form domain of the total Hamiltonian $H$. In Subsection \ref{subsection:regularity}, we establish regularity of $H$ with respect to $A$ and in Subsection \ref{subsection:Mourre}, we prove a Mourre estimate. Putting all together, we finally deduce in Subsection \ref{subsection:proofs} that the statements of Theorems \ref{thm:loc-spec} and \ref{thm:abs-cont} hold.

\subsection{The conjugate operator and its associated semigroup}\label{subsection:conjugate}

Let $a$ be the operator on $L^2( \mathbb{R}_+ )$ defined by the expression
\begin{equation}\label{eq:defa}
a = \frac{ i }{ 2 } \big ( f(p) \partial_p + \partial_p f( p ) \big ) = i f( p ) \partial_p + \frac{ i }{ 2 } f'(p),
\end{equation}
where $f(p) := p^{-1} \omega( p ) = p^{-1} \sqrt{ p^2 + \me^2 }$ and $f'$ stands for the derivative of $f$. The operator $a$ with domain $\mathrm{C}_0^\infty( ( 0 , \infty ) )$ is symmetric; its closure is denoted by the same symbol.

We construct the $\mathrm{C}_0$-semigroup, $w_t$, associated with $-a$. Let
\begin{equation*}
g(p) := \int_0^p \frac{1}{f(r)} \d r = \sqrt{ p^2 + \me^2 } - \me.
\end{equation*}
Note that the function $g$ is bijective on $[ 0 , \infty )$ with inverse $g^{-1}(p) = \sqrt{ ( p + \me)^2  - \me^2 }$. Let $\phi_t ( p ) = g^{-1} ( t + g (p ) )$ be the flow associated to the vector field $f(p) \partial_p$, i.e. $\partial_t \phi_t(p) = f(p) \partial_p \phi_t(p)$, $\phi_0(p) = p$. Setting
\begin{equation*}
(w_t u )(p) := \big | \partial_p \phi_t(p) \big |^{\frac12} u ( \phi_t(p) ) ,
\end{equation*}
one easily verifies that $w_t$ is the (contraction) $\mathrm{C}_0$-semigroup associated with $-a$, namely $w_{t+s} = w_t w_s$ for $t,s \ge 0$, and $( \partial_t w_t u ) |_{t = 0 } (p) = - i (a u)(p)$. We observe that $a$ is maximal symmetric with deficiency index $n_{-}=\dim \mathrm{Ker} (a^* + i) = 0$.

On $\mathfrak{H}_{c} = \oplus_\gamma L^2(\R _+)$, the operator $\oplus_\gamma a$ is still denoted by the symbol $a$. Our conjugate operator, $A$, acting on the full Hilbert space $\mathcal{H} = \mathfrak{F}_a \otimes \mathfrak{F}_a \otimes \mathfrak{F}_{Z^0}$, is then given by
\begin{equation}\label{eq:defA}
A := \d \Gamma( a ) \otimes \1 \otimes \1 + \1 \otimes \d \Gamma( a ) \otimes \1.
\end{equation}
From the properties of $a$, we deduce that $A$ is maximal symmetric and that $-A$ generates the $\mathrm{C}_0$-semigroup
\begin{equation*}
W_t := \Gamma( w_t ) \otimes \Gamma( w_t ) \otimes \1.
\end{equation*}

The adjoint semigroup, $W_t^*$, with generator $A^*$, is given as follows: Let $\psi_t $ be defined on $[ 0 , \infty )$ by $\psi_t(p) := 0$ if $p \le \sqrt{ (t + \me )^2 - \me^2 }$ and $\psi_t(p) := g^{-1}( - t + g(p) )$ otherwise. One can verify that the adjoint semigroup of $w_t$ is the $\mathrm{C}_0$-semigroup of isometries given by
\begin{equation*}
(w_t^* u )(p) = \big | \partial_p \psi_t(p) \big |^{\frac12} u ( \psi_t(p) ) .
\end{equation*}
We deduce that
\begin{equation*}
W_t^* = \Gamma( w_t^* ) \otimes \Gamma( w_t^* ) \otimes \1 ,
\end{equation*}
and that $W_t^*$ is a $\mathrm{C}_0$-semigroup of isometries on $\mathcal{H}$.

The form domain of $H$ is denoted by
\begin{equation*}
\mathcal{G} := \D( |H|^{\frac12} ) = \D( H_0^{\frac12} ).
\end{equation*}
\begin{prop}\label{prop:semigroup}
For all $t \ge 0$, we have that
\begin{equation*}
W_t \, \mathcal{G} \subset \mathcal{G} , \quad W_t^* \, \mathcal{G} \subset \mathcal{G} ,
\end{equation*}
and
\begin{equation*}
\big \| H_0^{\frac12} W_t ( H_0^{\frac12} + \1 )^{-1} \big \| \le 1, \quad \big \| H_0^{\frac12} W_t^* ( H_0^{\frac12} + \1 )^{-1} \big \| \le 1.
\end{equation*}
In particular, Hypothesis \ref{hyp:M1} of Appendix \ref{appendix:Mourre} is satisfied.
\end{prop}
\begin{proof}
We prove the statement for $W_t$, the proof for $W_t^*$ is similar. First, we show that $w_t \, \D( \omega ) \subset \D( \omega )$ and that
\begin{equation}\label{eq:a5}
\big \| \omega^{-\frac12} \, w_t^* \, \omega \, w_t  \, \omega^{-\frac12} \big \| \le 1,
\end{equation}
where, recall, $\omega$ is the multiplication operator by $\omega( p ) = \sqrt{ p^2 + \me^2 }$ on $L^2( \mathbb{R}_+ )$. For any $u \in \mathrm{C}_0^\infty( ( 0 , \infty ) )$, we have that
\begin{align*}
\| \omega \, w_t \, u \|^2 = \int \omega( p )^2 \big | \partial_p \phi_t( p ) \big |Ê\big | u ( \phi_t ( p ) ) \big |^2 \d p .
\end{align*}
Using the definition of $\phi_t$, one sees that $\phi_t(p) \ge p$ for all $t \ge 0$, and hence
\begin{align*}
\| \omega \, w_t \, u \|^2 \le \int \omega( \phi_t( p ) )^2 \big | \partial_p \phi_t( p ) \big |Ê\big | u ( \phi_t ( p ) ) \big |^2 \d p = \| \omega \, u \|^2.
\end{align*}
Since $\mathrm{C}_0^\infty( ( 0 , \infty ) )$ is a core for $\omega$, this implies that $w_t \, \D( \omega ) \subset \D( \omega )$ and that
\begin{equation*}
\big \| \omega \, w_t \, \omega^{-1} \big \| \le 1.
\end{equation*}
Using the fact that $w_t^*$ is isometric and an interpolation argument, we obtain \eqref{eq:a5}.

Now, let $\varphi \in \mathfrak{F}_{a , \mathrm{fin} }( \D( \omega ) ) \otimes \mathfrak{F}_{a , \mathrm{fin} }( \D( \omega ) ) \otimes \mathfrak{F}_{Z^0}$, where $\mathfrak{F}_{a , \mathrm{fin} }( \D( \omega ) )$ denotes the set of vectors $( \varphi_0 , \varphi_1 , \dots )$ in  $\oplus_{n=0}^\infty \, \otimes_a^n \D( \omega )$ (algebraic tensor product) such that $\varphi_n = 0$ for all but finitely many $n$'s. We compute
\begin{align*}
\big \| H_0^{\frac12} W_t \varphi \big \|^2 &= \big \langle \varphi , W_t^* H_0 W_t \varphi \big \rangle \\
& = \big \langle \varphi , \big ( \d \Gamma( w_t^* w_t , w_t^* \omega w_t ) \otimes \Gamma( w_t^* w_t ) \otimes \1 \\
& \qquad \qquad +  \Gamma( w_t^* w_t )  \otimes \d \Gamma( w_t^* w_t , w_t^* \omega w_t ) \otimes \1 \\
& \qquad \qquad +  \Gamma( w_t^* w_t ) \otimes \Gamma( w_t^* w_t ) \otimes \d \Gamma( \omega_3 ) \big ) \varphi \big \rangle ,
\end{align*}
where, for $c_1$, $c_2$ operators on $\mathfrak{H}_c$, the operator $\d \Gamma( c_1 , c_2 )$ on $\mathfrak{F}_a$ is defined by (see \cite{ammari,DG})
  \begin{align*}
& \d \Gamma( c_1, c_2 ) \Omega_a = 0 , \\
& \d \Gamma( c_1, c_2 )\vert_{  {\otimes_a^n} \mathfrak{H}_c } = \sum_{j=1}^{n} \underbrace{c_1 \otimes \cdots \otimes c_1 }_{j-1} \otimes c_2 \otimes \underbrace{c_1 \otimes \cdots \otimes c_1 }_{n - j}.
\end{align*}
Combining \eqref{eq:a5}, the bound $\| w_t^* w_t \| \le 1$, and \cite[Lemma 2.3]{ammari} (see also \cite[Lemma 2.8]{DG}), we obtain
\begin{align*}
\big \| H_0^{\frac12} W_t \varphi \big \|^2 &\le  \big \| ( \d \Gamma( \omega )^{\frac12}  \otimes \1 \otimes \1 \big ) \varphi \big \|^2 + \big \| (  \1 \otimes \d \Gamma( \omega )^{\frac12}  \otimes  \1 \big ) \varphi \big \|^2 \\
&\quad + \big \| (  \1 \otimes \1 \otimes \d \Gamma( \omega_3 )^{\frac12}  \big ) \varphi \big \|^2 \\
& = \big \| H_0^{\frac12} \varphi \big \|^2.
\end{align*}
This concludes the proof.
\end{proof}

\subsection{Regularity of the Hamiltonian with respect to the conjugate operator}\label{subsection:regularity}

Recall that the conjugate operator $A$ is defined by the expressions \eqref{eq:defa} and \eqref{eq:defA}. In this subsection, we prove the following proposition.
\begin{prop}\label{thm:interpo}
Assume that Hypothesis \ref{hyp:2} holds. Let $|g| \ll \me$. Then we have that
\begin{equation*}
H \in \mathrm{C}^{1,1}( A_{ \mathcal{G} } ; A_{\mathcal{G}^*} ) ,
\end{equation*}
in the sense of Hypothesis \ref{hyp:M3} of Appendix \ref{appendix:Mourre}.
\end{prop}
To prove Proposition \ref{thm:interpo}, we use real interpolation. We have that
\begin{equation*}
[ H_0 , i A ] = \Ne \otimes \1 \otimes \1 + \1 \otimes \Np \otimes \1 ,
\end{equation*}
in the sense of quadratic forms on $\D( H_0 ) \cap \D( A )$. Since $\D( H_0 ) \cap \D( A )$ is a core for $H_0$ and since $\Ne \otimes \1 \otimes \1 + \1 \otimes \Np \otimes \1$ is relatively $H_0$-bounded, Proposition \ref{prop:semigroup} together with Proposition \ref{prop:B3} imply that $H_0$ belongs to $\mathrm{C}^{1}( A_{ \mathcal{G} } ; A_{\mathcal{G}^*} )$. Next, since $[ H_0 , i A ]$ commutes with $A$, we easily deduce that $H_0 \in \mathrm{C}^{2}( A_{ \mathcal{G} } ; A_{\mathcal{G}^*} )$, and hence in particular $H_0 \in \mathrm{C}^{1,1}( A_{ \mathcal{G} } ; A_{\mathcal{G}^*} )$.
 Here we recall that, for all $0 \le \theta \le 1$ and $1 \le q < \infty$,
\begin{align}\label{def:Cthetaq}
\mathrm{C}^{\theta , q}( A_{ \mathcal{G} } ; A_{\mathcal{G}^*} ) := \Big \{ & T \in \mathcal{B}( \mathcal{G} ; \mathcal{G}^* ) , W_t^* T W_t - T \in \mathcal{B}( \mathcal{G} ; \mathcal{G}^* ) \text{ for all } t \in (0,1) ,  \notag \\
& \int_0^1 t^{-\theta q - 1 } \big \| W_t^* T W_t - T \big \|^q_{ \mathcal{B}( \mathcal{G} ; \mathcal{G}^* ) } \d t < \infty \Big \} .
\end{align}

In order to prove that $H \in \mathrm{C}^{1,1}( A_{ \mathcal{G} } ; A_{\mathcal{G}^*} )$, it remains to show that the interaction Hamiltonian $H_I \in \mathrm{C}^{1,1}( A_{ \mathcal{G} } ; A_{\mathcal{G}^*} )$. Using in particular Proposition \ref{prop:B3}, we see that it suffices, in fact,  to verify that the commutator $[ H_I , i A ]$ belongs to $\mathcal{B} ( \mathcal{G} ; \mathcal{G}^* )$ and that $[ H_I , i A ] \in \mathrm{C}^{0,1}( A_{ \mathcal{G} } ; A_{\mathcal{G}^*} )$. This is the purpose of the remainder of this section.

We use the notation \eqref{eq:interaction-12c}. Using Hypothesis \ref{hyp:2} and the estimates of Appendix \ref{appendixA} (see \eqref{eq:deriv-estimate-1}--\eqref{eq:deriv-estimate-2} and \eqref{eq:deriv2nd-estimate-1}--\eqref{eq:deriv2nd-estimate-4}), we can rewrite
\begin{align}
F^{(\alpha)}(\xi_1,\, \xi_2,\, \xi_3) := \tilde h^{(\alpha)}( \xi_1 , \xi_2 , \xi_3 ) \tilde G^{(\alpha)}( \xi_1 , \xi_2 , \xi_3 )   , \label{eq:interaction-12new}
\end{align}
where $\tilde h^{(\alpha)}( \xi_1 , \xi_2 , \xi_3 )$ is of the form
\begin{align}
\tilde h^{(\alpha)}( \xi_1 , \xi_2 , \xi_3 ) = p_1 p_2 s^{(\alpha)}( \xi_1 , \xi_2 , \xi_3 ) ,
\end{align}
with $s^{( \alpha )}$ satisfying, for all $n,m \in \{ 0 , 1 , 2 \}$,
\begin{align}\label{eq:estimates_s}
\big | \partial^n_{p_1} \partial_{p_2}^m s^{(\alpha)} ( \xi_1 , \xi_2 , \xi_3 ) \big | \lesssim p_1^{-n} p_2^{-m} ,
\end{align}
in a neighborhood of $0$.

 Moreover the kernels $\tilde G^{(\alpha)}$ satisfy
\begin{itemize}
\item[($a_0$)] There exists a compact set $K \subset \mathbb{R}_+ \times \mathbb{R}_+ \times \mathbb{R}^3$ such that $\tilde G^{(\alpha)} ( p_1 , \gamma_1 , p_2 , \gamma_2 , k , \lambda ) = 0$ if $( p_1 , p_2 , k ) \notin K$.
\item[($b_0$)] There exists $\varepsilon > 0$ such that
\begin{align*}
\sum_{ \gamma_1 , \gamma_2 , \lambda } \int ( 1 + x_1^2 + x_2^2 )^{1 + \varepsilon } \left | \hat{ \tilde G}^{(\alpha)}( x_1 , \gamma_1 , x_2 , \gamma_2 , k , \lambda )\right|^2 \d x_1 \d x_2 \d k < \infty ,
\end{align*}
where, recall, $\hat{ \tilde G}^{(\alpha)}$ denote the Fourier transform of $\tilde G^{(\alpha)}$ with respect to the variables $(p_1,p_2)$, and $x_l$, $l=1,2$, is the variable dual to $p_l$.
\item[($c_0$)] If $p_1 = 0$ or $p_2 = 0$, then $\tilde G^{(\alpha)} ( p_1 , \gamma_1 , p_2 , \gamma_2 , k , \lambda ) = 0$.
\end{itemize}

Our strategy consists in working with interaction operators of the form \eqref{eq:interaction} with $H_I^{(1)}$, $H_I^{(2)}$ given by \eqref{eq:interaction-1b}--\eqref{eq:interaction-2b} and $F^{(1)}$, $F^{(2)}$ satisfying \eqref{eq:interaction-12new}--\eqref{eq:estimates_s}. We then use an interpolation argument for the kernels $\tilde G^{(\alpha)}$.
\begin{lem}\label{lem:interpo}
Consider the operator $H_I$ of the form \eqref{eq:interaction} with $H_I^{(1)}$, $H_I^{(2)}$ given by \eqref{eq:interaction-1b}--\eqref{eq:interaction-2b} and $F^{(1)}$, $F^{(2)}$ satisfying \eqref{eq:interaction-12new}--\eqref{eq:estimates_s}.
\begin{itemize}
\item[(i)] Suppose that $\tilde G^{(\alpha)} \in L^2( \Sigma \times \Sigma \times \Sigma_3 )$ satisfy the following conditions
\begin{itemize}
\item[(i)(a)] There exists a compact set $K \subset \mathbb{R}_+ \times \mathbb{R}_+ \times \mathbb{R}^3$ such that $\tilde G^{(\alpha)} ( p_1 , \gamma_1 , p_2 , \gamma_2 , k , \lambda ) = 0$ if $( p_1 , p_2 , k ) \notin K$.
\item[(i)(b)]
\begin{align*}
\sum_{ \gamma_1 , \gamma_2 , \lambda } \int ( 1 + x_1^2 + x_2^2 ) \left | \hat{ \tilde G}^{(\alpha)}( x_1 , \gamma_1 , x_2 , \gamma_2 , k , \lambda )\right|^2 \d x_1 \d x_2 \d k < \infty .
\end{align*}
\item[(i)(c)] If $p_1 = 0$ or $p_2 = 0$, then $\tilde G^{(\alpha)} ( p_1 , \gamma_1 , p_2 , \gamma_2 , k , \lambda ) = 0$.
\end{itemize}
Then $H'_I = [ H_I , i A ] \in \mathrm{C}^{0}( A_{ \mathcal{G} } ; A_{\mathcal{G}^*} ) \equiv \mathcal{B} ( \mathcal{G} ; \mathcal{G}^* )$.
\item[(ii)] Suppose that $\tilde G^{(\alpha)} \in L^2( \Sigma \times \Sigma \times \Sigma_3 )$ satisfy the following conditions
\begin{itemize}
\item[(ii)(a)] There exists a compact set $K \subset \mathbb{R}_+ \times \mathbb{R}_+ \times \mathbb{R}^3$ such that $\tilde G^{(\alpha)} ( p_1 , \gamma_1 , p_2 , \gamma_2 , k , \lambda ) = 0$ if $( p_1 , p_2 , k ) \notin K$.
\item[(ii)(b)]
\begin{align*}
\sum_{ \gamma_1 , \gamma_2 , \lambda } \int ( 1 + x_1^2 + x_2^2 )^3 \left | \hat{ \tilde G}^{(\alpha)}( x_1 , \gamma_1 , x_2 , \gamma_2 , k , \lambda )\right|^2 \d x_1 \d x_2 \d k < \infty .
\end{align*}
\item[(ii)(c)] If $p_1 = 0$ or $p_2 = 0$, then $D^\beta \tilde G^{(\alpha)} ( p_1 , \gamma_1 , p_2 , \gamma_2 , k , \lambda ) = 0$ for all multi-index $\beta = ( \beta_1 , \beta_2 )$, $| \beta | \le 2$, with $D^\beta = \partial^{\beta_1+\beta_2 } / \partial_{x_1^{\beta_1}} \partial_{x_2^{\beta_2}}$.
\end{itemize}
Then $H'_I = [ H_I , i A ] \in \mathrm{C}^{1}( A_{ \mathcal{G} } ; A_{\mathcal{G}^*} )$.
\end{itemize}
\end{lem}
\begin{proof}
$(i)$ Recall that the conjugate operator $A$ is defined by Eq. \eqref{eq:defA}, with
\begin{equation*}
a = i f( p ) \partial_p + \frac{ i }{ 2 } f'(p) ,
\end{equation*}
and $f(p) = p^{-1} \sqrt{ p^2 + \me^2 }$. We use the notation $a_l = i f( p_l ) \partial_{p_l } + \frac{ i }{ 2 } f'(p_l)$, for $l=1,2$.  We then have that
\begin{equation}\label{eq:comm-HI}
 [H_I,\, iA] = H_I(-i a_{1} F) + H_I(-ia_{2} F)\, ,
\end{equation}
in the sense of quadratic forms on $\mathfrak{D}( H_0 ) \cap \mathfrak{D}( A )$.

Recalling the notations $\xi_l = ( p_l , \gamma_l )$, we compute
\begin{align}
( a_1 F^{(\alpha)} ) ( \xi_1 , \xi_2 , \xi_3 ) =& \Big [ \frac{i}{2} p_1 p_2 f'(p_1) s^{(\alpha)}( \xi_1 , \xi_2 , \xi_3 ) + i p_2 f( p_1 ) s^{(\alpha)}( \xi_1 , \xi_2 , \xi_3 )  \notag \\
& + i p_1 p_2 f(p_1) ( \partial_{p_1} s^{(\alpha)} ) ( \xi_1 , \xi_2 , \xi_3 ) \Big ] \tilde G^{(\alpha)}( \xi_1 , \xi_2 , \xi_3 ) \notag \\
& + i p_1 p_2 f( p_1 ) s^{(\alpha)}( \xi_1 , \xi_2 , \xi_3 ) ( \partial_{p_1} \tilde G^{(\alpha)} )( \xi_1 , \xi_2 , \xi_3 ) .  \label{eq:aaa1}
\end{align}
Using \eqref{eq:estimates_s} and the definition of $f$, we see that the term in brackets satisfy
\begin{align*}
& \Big | \frac{i}{2} p_1 p_2 f'(p_1) s^{(\alpha)}( \xi_1 , \xi_2 , \xi_3 ) + i p_2 f( p_1 ) s^{(\alpha)}( \xi_1 , \xi_2 , \xi_3 ) + i p_1 p_2 f(p_1) ( \partial_{p_1} s^{(\alpha)} ) ( \xi_1 , \xi_2 , \xi_3 ) \Big | \\
& \lesssim p_1^{-1} p_2 ,
\end{align*}
in any compact set. Now, since $p_1 \mapsto \tilde G^{(\alpha)}( p_1 , \gamma_1 , \xi_2 , \xi_3 ) \in H_0^{1} ( \mathbb{R}_+ )$ by the onditions $(i)(b)$ and $(i)(c)$, and since $\tilde G^{(\alpha)}$ is compactly supported in the variables $(p_1,p_2,k)$ by the condition $(i)(a)$, we deduce that
\begin{align*}
p_1^{-1} p_2 \tilde G^{(\alpha)}( \xi_1 , \xi_2 , \xi_3 ) \in L^2( \d \xi_1 \d \xi_2 \d \xi_3 ).
\end{align*}
Here we used that
\begin{align*}
\big \| p_1^{-1} p_2 \tilde G^{(\alpha)}( \xi_1 , \xi_2 , \xi_3 ) \big \|_{ L^2( \d \xi_1 \d \xi_2 \d \xi_3 ) } \lesssim \big \| p_2 \partial_{p_1} \tilde G^{(\alpha)}( \xi_1 , \xi_2 , \xi_3 ) \big \|_{ L^2( \d \xi_1 \d \xi_2 \d \xi_3 ) } ,
\end{align*}
by Hardy's inequality at the origin in $H_0^{1} (\mathbb{R}_+ )$. Likewise, we have that
\begin{equation*}
\big | i p_1 p_2 f( p_1 ) s^{(\alpha)}( \xi_1 , \xi_2 , \xi_3 ) \big | \lesssim 1 ,
\end{equation*}
in any compact set, and hence, using again that $p_1 \mapsto \tilde G^{(\alpha)}( p_1 , \gamma_1 , \xi_2 , \xi_3 ) \in H_0^{1} ( \mathbb{R}_+ )$ and that $\tilde G^{(\alpha)}$ is compactly supported in the variables $(p_1,p_2,k)$, it follows that
\begin{align*}
i p_1 p_2 f( p_1 ) s^{(\alpha)}( \xi_1 , \xi_2 , \xi_3 ) ( \partial_{p_1} \tilde G^{(\alpha)} )( \xi_1 , \xi_2 , \xi_3 ) \in L^2( \d \xi_1 \d \xi_2 \d \xi_3 ).
\end{align*}
The previous estimates show that
\begin{align*}
( a_1 F^{(\alpha)} ) ( \xi_1 , \xi_2 , \xi_3 ) \in L^2( \d \xi_1 \d \xi_2 \d \xi_3 ) ,
\end{align*}
and proceeding in the same way, one verifies that $( a_2 F^{(\alpha)} ) ( \xi_1 , \xi_2 , \xi_3 ) \in L^2( \d \xi_1 \d \xi_2 \d \xi_3 )$. Using the expression \eqref{eq:comm-HI} of the commutator $[ H_I , i A ]$ and the $N_\tau$ estimates of \cite{GlimmJaffe}, we immediately deduce that $[ H_I , i A ] \in \mathcal{B}( \mathcal{G} ; \mathcal{G}^* ) = \mathrm{C}^{0}( A_{ \mathcal{G} } ; A_{\mathcal{G}^*} )$.

$(ii)$ It suffices to proceed similarly. More precisely, we compute the second commutator
\begin{equation}\label{eq:second-comm-HI}
\big [ \big [ H_I , i A \big ] , i A \big ] = - H_I(a_{1}^2 F) - H_I( a_{2}^2 F) - 2 H_I ( a_1a_2 F ) .
\end{equation}
Computing $a_1^2 F$, $a_2^2 F$ and $a_1 a_2 F$ yields to several terms that are estimated separately. Each term, however, can be treated in the same way, using Hardy's inequality together with the assumptions $(ii)(a)$, $(ii)(b)$, $(ii)(c)$. We give an example. Consider the first term inside the brackets of \eqref{eq:aaa1} and apply to it the operator $i f( p_1 ) \partial_{p_1}$. This gives in particular a term of the form
\begin{align*}
- \frac{1}{2} p_2 f( p_1 ) f'(p_1) s^{(\alpha)}( \xi_1 , \xi_2 , \xi_3 ) \tilde G^{(\alpha)}( \xi_1 , \xi_2 , \xi_3 ) ,
\end{align*}
that will appear in the expression of $a_1^2 F$. From \eqref{eq:estimates_s} and the definition of $f$,  it follows that
\begin{align*}
\big | p_2 f( p_1 ) f'(p_1) s^{(\alpha)}( \xi_1 , \xi_2 , \xi_3 ) \tilde G^{(\alpha)}( \xi_1 , \xi_2 , \xi_3 )  \big | \lesssim p_1^{-3} p_2 \big | \tilde G^{(\alpha)}( \xi_1 , \xi_2 , \xi_3 )  \big |,
\end{align*}
in any compact set. Since $p_1 \mapsto \tilde G^{(\alpha)}( p_1 , \gamma_1 , \xi_2 , \xi_3 ) \in H_0^{3} ( \mathbb{R}_+ )$ by the conditions $(ii)(b)$ and $(ii)(c)$, and since $\tilde G^{(\alpha)}$ is compactly supported in the variables $(p_1,p_2,k)$ by the condition $(ii)(a)$, we obtain as above that
\begin{align*}
p_2 f( p_1 ) f'(p_1) s^{(\alpha)}( \xi_1 , \xi_2 , \xi_3 ) \tilde G^{(\alpha)}( \xi_1 , \xi_2 , \xi_3 ) \in L^2( \d \xi_1 \d \xi_2 \d \xi_3 ).
\end{align*}
Here we used that
\begin{align*}
\big \| p_1^{-3} p_2 \tilde G^{(\alpha)}( \xi_1 , \xi_2 , \xi_3 ) \big \|_{ L^2( \d \xi_1 \d \xi_2 \d \xi_3 ) } \lesssim \big \| p_2 \partial_{p_1}^3 \tilde G^{(\alpha)}( \xi_1 , \xi_2 , \xi_3 ) \big \|_{ L^2( \d \xi_1 \d \xi_2 \d \xi_3 ) } ,
\end{align*}
by Hardy's inequality at the origin in $H_0^{3} (\mathbb{R}_+ )$.
Treating all the other terms in a similar manner, we deduce that
\begin{equation*}
a_{1}^2 F + a_{2}^2 F + 2 a_1a_2 F \in L^2( \d \xi_1 \d \xi_2 \d \xi_3 ) ,
\end{equation*}
and therefore that $[ [ H_I , i A ] , i A ] \in \mathcal{B}( \mathcal{G} ; \mathcal{G}^* )$. Together with Proposition \ref{prop:semigroup}, this shows $(ii)$.
\end{proof}

\begin{proof}[Proof of Proposition \ref{thm:interpo}]
By the observation after the statement of Proposition \ref{thm:interpo}, we already now that $H_0 \in \mathrm{C}^{1,1}( A_{ \mathcal{G} } ; A_{\mathcal{G}^*} ) $. Hence, to conclude the proof of Proposition \ref{thm:interpo}, it suffices to verify that $H_I \in \mathrm{C}^{1,1}( A_{ \mathcal{G} } ; A_{\mathcal{G}^*} )$. Recall that $H_I$ is the sum of $4$ terms, see \eqref{eq:interaction}. We consider for instance the first one, $H_I^{(1)}$. The other terms can be treated in the same way.

Let $K_0 \subset \mathbb{R}_+ \times \mathbb{R}_+ \times \mathbb{R}^3$ be a compact set. Let $S_{(i)}$ denote the set of all $\tilde G^{(1)} \in L^2( \Sigma \times \Sigma \times \Sigma_3 )$ satisfying the conditions $(i)(a)$ (with $K=K_0$), $(i)(b)$ and $(i)(c)$, equipped with the norm
\begin{align*}
\big \| \tilde G^{(1)} \big \|_{S_{(i)}} := \sum_{ \gamma_1 , \gamma_2 , \lambda } \int ( 1 + x_1^2 + x_2^2 ) \left | \hat{ \tilde G}^{(1)}( x_1 , \gamma_1 , x_2 , \gamma_2 , k , \lambda )\right|^2 \d x_1 \d x_2 \d k  .
\end{align*}
Likewise, we denote by $S_{(ii)}$ the set of all $\tilde G^{(1)} \in L^2( \Sigma \times \Sigma \times \Sigma_3 )$ satisfying the conditions $(ii)(a)$ (with $K=K_0$), $(ii)(b)$ and $(ii)(c)$, equipped with the norm
\begin{align*}
\big \| \tilde G^{(1)} \big \|_{S_{(ii)}} := \sum_{ \gamma_1 , \gamma_2 , \lambda } \int ( 1 + x_1^2 + x_2^2 )^3 \left | \hat{ \tilde G}^{(1)}( x_1 , \gamma_1 , x_2 , \gamma_2 , k , \lambda )\right|^2 \d x_1 \d x_2 \d k  .
\end{align*}
By Lemma \ref{lem:interpo} and its proof, the map
\begin{align}
S_{(i)} \ni \tilde G^{(1)} \mapsto H^{(1)'}_I( \tilde h^{(1)} \tilde G^{(1)} )  \in \mathrm{C}^{0}( A_{ \mathcal{G} } ; A_{\mathcal{G}^*} )
\end{align}
is linear and continuous, and, likewise, the map
\begin{align}
S_{(ii)} \ni \tilde G^{(1)} \mapsto H^{(1)'}_I( \tilde h^{(1)} \tilde G^{(1)} ) \in \mathrm{C}^{1}( A_{ \mathcal{G} } ; A_{\mathcal{G}^*} )
\end{align}
is linear and continuous. Here we have used the notation
\begin{equation*}
H^{(1)'}_I( \tilde h^{(1)} \tilde G^{(1)} ) := [ H^{(1)}_I( \tilde h^{(1)} \tilde G^{(1)} ) , i A ].
\end{equation*}
By real interpolation, we deduce that
\begin{align}
\big ( S_{(i)} , S_{(ii)} \big )_{\theta,2}  \ni \tilde G^{(1)} \mapsto H^{(1)'}_I( \tilde h^{(1)} \tilde G^{(1)} ) \in \big ( \mathrm{C}^{0}( A_{ \mathcal{G} } ; A_{\mathcal{G}^*} ) , \mathrm{C}^{1}( A_{ \mathcal{G} } ; A_{\mathcal{G}^*} ) \big )_{ \theta , 2 } ,
\end{align}
for all $0 \le \theta \le 1$.

Now, by \cite[Section 5]{ABG}, we have that
\begin{align}
\big ( \mathrm{C}^{0}( A_{ \mathcal{G} } ; A_{\mathcal{G}^*} ) , \mathrm{C}^{1}( A_{ \mathcal{G} } ; A_{\mathcal{G}^*} ) \big )_{ \theta , 2 } = \mathrm{C}^{\theta , 2}( A_{ \mathcal{G} } ; A_{\mathcal{G}^*} ) ,
\end{align}
for all $0 < \theta < 1$, and using the definition \eqref{def:Cthetaq}, one easily verifies that
\begin{equation*}
\mathrm{C}^{\theta , 2}( A_{ \mathcal{G} } ; A_{\mathcal{G}^*} ) \subset \mathrm{C}^{ 0 , 1 }( A_{ \mathcal{G} } ; A_{\mathcal{G}^*} ).
\end{equation*}
On the other hand, from the definition of the interpolated space $\big ( S_{(i)} , S_{(ii)} \big )_{\theta,2}$ and mimicking the method allowing one to compute the interpolation of Sobolev spaces (see e.g. \cite{Tr}), it is not difficult to verify that, for $0 < \varepsilon < 2 \theta < 1$, the set of all kernels $\tilde G^{(1)} \in L^2( \Sigma \times \Sigma \times \Sigma_3 )$ satisfying the conditions $(a_0)$, $(b_0)$ and $(c_0)$ stated above is included in $\big ( S_{(i)} , S_{(ii)} \big )_{\theta,2}$. This shows, in particular, that $H_I^{(1)'} \in \mathrm{C}^{0 , 1}( A_{ \mathcal{G} } ; A_{\mathcal{G}^*} )$, and hence that $H_I^{(1)} \in \mathrm{C}^{1 , 1}( A_{ \mathcal{G} } ; A_{\mathcal{G}^*} )$. Since the other terms, ${H_I^{(1)}}^*$, $H_I^{(2)}$ and ${H_I^{(2)}}^*$, can be treated in the same way, this concludes the proof.
\end{proof}

\subsection{The Mourre estimate}\label{subsection:Mourre}

Given $\tilde{F} = (\tilde{F}^{(1)}, \tilde{F}^{(2)})\in ( \mathfrak{H}_c\otimes\mathfrak{H}_c \otimes L^2(\Sigma_3) )^2$, and for $H_I^{(i)}(\tilde{F}^{(i)})$ given by \eqref{eq:interaction-1b}-\eqref{eq:interaction-2b}, we define
\begin{equation}\nonumber
 H_I(\tilde F) = H_I^{(1)}(\tilde{F}^{(1)}) + (H_I^{(1)}(\tilde{F}^{(1)}))^*
 + H_I^{(2)}(\tilde{F}^{(2)}) + (H_I^{(2)}(\tilde{F}^{(2)}))^*\, .
\end{equation}

\begin{prop}\label{prop:Mourre}
Assume that Hypothesis \ref{hyp:2} hold and let $\delta\in (0, \, \me)$. There exist $g_\delta>0$, $\mathrm{c}_\delta>0$ and $\mathrm{C} \in \mathbb{R}$ such that, for all $|g|\leq g_\delta$, and for
$$
\Delta := [ \delta , \mz - \delta ] ,
$$
we have, in the sense of quadratic forms on $\mathcal{D}(A)\cap\mathcal{D}(H_0)$,
\begin{equation}\label{eq:main-mourre}
 [H, iA]  \geq \mathrm{c}_\delta \1 - \mathrm{C} \1_{\Delta}^\perp(H-E) \langle H \rangle  ,
\end{equation}
where we have set $E := \inf \mathrm{spec}(H)$, $\1_{\Delta}^\perp(H-E) := \1 - \1_{\Delta}(H-E)$ and $\langle H \rangle := ( \1 + H^2 )^{1/2}$.
\end{prop}
\begin{proof}[Proof of Proposition \ref{prop:Mourre}]
As in Subsection \ref{subsection:regularity}, we have, in the sense of quadratic forms on $\mathcal{D}(A)\cap\mathcal{D}(H_0)$,
\begin{equation}\label{eq:mourre-1}
\begin{split}
 [H_0,\, iA] = N_+\otimes \1\otimes \1 + \1\otimes N_- \otimes \1\, ,
\end{split}
\end{equation}
where $\Ne$ (respectively $\Np$) is the number operator for electrons (respectively po\-si\-trons) as defined in \eqref{eq:def-NpNe}. In the sequel, by abuse of notation, we shall omit the identity operators in $N_+\otimes \1\otimes \1$ and $\1\otimes N_- \otimes\1$ and denote them respectively again by $\Ne$ and $\Np$.

Let $\ael = a\otimes\1\otimes\1$ be the conjugate operator for electron acting on the $p_1$ variable in $\mathfrak{H}_c\otimes\mathfrak{H}_c\otimes L^2(\Sigma_{3})$ and $\ap=\1\otimes a \otimes\1$ be the conjugate operator for positron acting on the $p_2$ variable. As in \eqref{eq:comm-HI}, we have that
\begin{equation}\label{eq:mourre-2}
 [H_I,\, iA] = H_I(-i a_{1} F) + H_I(-ia_{2} F)\, ,
\end{equation}
in the sense of quadratic forms on $\mathcal{D}(A)\cap\mathcal{D}(H_0)$. Here we recall that $a_1 F$ and $a_2 F$ belong to $L^2 ( \mathrm{d} \xi_1 \mathrm{d} \xi_2 \mathrm{d} \xi_3 )$ as follows from the estimates of Appendix \ref{appendixA} and Hypothesis \ref{hyp:2} (see more precisely the proof of Lemma \ref{lem:interpo} $(i)$).

For $P_{\Omega_a\times\Omega_a}:= P_{\Omega_a}\otimes P_{\Omega_a}\otimes \1$ being the projection onto the electron/positron vacuum, we have that
\begin{equation}\label{eq:mourre-3}
 \Ne + \Np + P_{\Omega_a\times\Omega_a} \geq \1 .
\end{equation}
Since $H = H_0 + g H_I$, and for $E= \inf \mathrm{spec}(H)$, we obtain from \eqref{eq:mourre-1}-\eqref{eq:mourre-2} that
\begin{align}
  [H,\, iA]  &=  \left( \Ne + \Np + P_{\Omega_a\times\Omega_a} \right)  -  P_{\Omega_a\times\Omega_a} + g  \left( H_I (-i \ael F) + H_I(-i\ap F)\right)  \notag \\
 &\ge \1  -  P_{\Omega_a\times\Omega_a} + g  \left( H_I (-i \ael F) + H_I(-i\ap F)\right) , \label{eq:mourre-5}
\end{align}
where we used the operator inequality \eqref{eq:mourre-3} in the last inequality. We estimate separately the two remainder terms occuring in the right hand side of \eqref{eq:mourre-5}.

Let us define a function $f_\Delta \in \mathrm{C}_0^\infty(\R)$ such that $0\leq f_\Delta \leq 1$ and
\begin{equation}
 f_\Delta (\lambda) =
    \left\{
    \begin{array}{ll}
        1 & \mbox{ if } \lambda\in [ \delta , \mz - \delta ] , \\
        0 & \mbox{ if } \lambda< \delta / 2 \mbox{ or } \lambda > \mz - \delta/2 .
    \end{array}
    \right.
\end{equation}
We observe that
\begin{equation}\label{eq:mourre-4}
 P_{\Omega_a\times\Omega_a} \, f_\Delta(H_0) = 0\, .
\end{equation}
The last identity holds because $ P_{\Omega_a\times\Omega_a}$ is a projection commuting with $H_0$ and because $\mathrm{supp}(f_\Delta) \cap \mathrm{spec}( H_0  P_{\Omega_a\times\Omega_a} ) = \emptyset$.
As in the proof of Lemma \ref{lm:Hext}, let $\tilde f \in \mathrm{C}_0^\infty( \mathbb{C} )$ denote an almost analytic extension of $f_\Delta$ satisfying $\tilde f |_{ \mathbb{R} } = f_\Delta$ and $| \partial_{ \bar z } \tilde f( z ) | \le C_n | \mathrm{Im} \, z |^n$ for any $n \in \mathbb{N}$.
 Thus, for $\d\tilde{f}(z) := -\frac{1}{\pi} \frac{\partial \tilde{f}}{\partial \overline{z}} (z) \, \d \mathrm{Re} \, z \, \d \mathrm{Im} \, z$, using Helffer-Sj\"ostrand functional calculus and the second resolvent equation, we obtain
\begin{align}
   f_\Delta(H-E) - f_\Delta(H_0)&
   = \int (H-E-z)^{-1} (H-E-H_0) (H_0-z)^{-1} \, \d\tilde{f}(z) \notag \\
   & = \int (H-E-z)^{-1} g H_I(F) (H_0 - z)^{-1} \, \d\tilde{f}(z) \notag \\
   & \ \ \ - E\int  (H-E-z)^{-1} (H_0 - z)^{-1} \, \d\tilde{f}(z)\, .\label{eq:mourre-10}
\end{align}
From Corollary~\ref{cor-rel-bound}, since Hypothesis~\ref{hyp:1} holds, there exists a constant $C$ such that
\begin{equation}\label{eq:mourre-12}
 \| H_I(F) (H_0 + 1)^{-1} \| \leq C  K(G)\, ,
\end{equation}
where $h^{(\alpha)} G^{(\alpha)} =  F^{(\alpha)}$ (see \eqref{eq:interaction-12c}) and $K(G)=K_2(G)$ is given by \eqref{ed:def-K-G-alpha} and \eqref{eq:def-kig}.

Therefore, with the inequality
\begin{equation}\label{eq:mourre-11}
\big \| (H_0+1) (H_0 -z)^{-1} \big \| \leq 1 + \frac{1 + |z|}{|\mathrm{Im} \, z|}\, ,
\end{equation}
and the properties of $\tilde{f}$, we obtain that there exists a constant $C_1>0$ depending only on $f_\Delta$ and $K(G)$ such that
\begin{align}
  & \left\| \int (H-E-z)^{-1} g H_I(F) (H_0 - z)^{-1} \, \d\tilde{f}(z)
  \right\| \notag \\
  & \leq |g| \int (1 + \frac{1 + |z|}{|\mathrm{Im} \,z|})\,  \| (H-E-z)^{-1} \|\,
  \| H_I(G) (H_0 + 1)^{-1} \| \, \d\tilde{f}(z) \leq  C_1 \, |g| \, . \label{eq:mourre-13}
\end{align}

Moreover, using again \eqref{eq:mourre-12}, standard perturbation theory yields that there exists $g_1>0$ such that for all $|g|\leq g_1$, we have
\begin{equation}\label{eq:mourre-14}
 |E| \leq |g|\, \frac{K(G)B_{\beta\eta}}{1 - g_1 K(G) C_{\beta \eta}}\, ,
\end{equation}
where $B_{\beta\eta}$ and $C_{\beta\eta}$ are the positive constants defined in Subsection \ref{subsection:self-adj}. Thus, there exists a constant $C_2$ depending on $f_\Delta$ and $K(G)$ such that
\begin{equation}\label{eq:mourre-15}
\Big \| E \int  (H-E-z)^{-1} (H_0 - z)^{-1} \, \d\tilde{f}(z) \Big \| \leq C_2 |g|\, .
\end{equation}
Inequalities \eqref{eq:mourre-10}, \eqref{eq:mourre-13} and \eqref{eq:mourre-15} give
\begin{equation}
 \| f_\Delta(H-E) - f_\Delta(H_0) \| \leq (C_1 +C_2)\, |g| . \label{eq:Mourre-18}
\end{equation}

For shortness, let $\1_\Delta \equiv \1_\Delta( H - E )$ and $\1^\perp_\Delta \equiv \1_\Delta^\perp( H - E )$. We have that
\begin{align}
- P_{\Omega_a\times\Omega_a} &=  - \1_\Delta P_{\Omega_a\times\Omega_a} \1_\Delta - \1_\Delta P_{\Omega_a\times\Omega_a} \1_\Delta^\perp - \1_\Delta^\perp P_{\Omega_a\times\Omega_a} \1_\Delta - \1_\Delta^\perp P_{\Omega_a\times\Omega_a} \1_\Delta^\perp \notag \\
&\ge  - \1_\Delta P_{\Omega_a\times\Omega_a} \1_\Delta - \1_\Delta P_{\Omega_a\times\Omega_a} \1_\Delta^\perp - \1_\Delta^\perp P_{\Omega_a\times\Omega_a} \1_\Delta - \1_\Delta^\perp . \label{eq:Mourre-19}
\end{align}
Using \eqref{eq:mourre-4} and \eqref{eq:Mourre-18}, we obtain that
\begin{align*}
\big \|Ê\1_\Delta P_{\Omega_a\times\Omega_a} \big \|Ê&\le \big \|Êf_\Delta ( H - E ) P_{\Omega_a\times\Omega_a} \big \| \\
& = \big \|Ê\big ( f_\Delta ( H - E ) - f_\Delta( H_0 ) \big ) P_{\Omega_a\times\Omega_a} \big \| \le (C_1 +C_2)\, |g| ,
\end{align*}
from which we deduce that
\begin{align*}
- \1_\Delta P_{\Omega_a\times\Omega_a} \1_\Delta - \1_\Delta P_{\Omega_a\times\Omega_a} \1_\Delta^\perp - \1_\Delta^\perp P_{\Omega_a\times\Omega_a} \1_\Delta \ge - 3 (C_1 +C_2)\, |g| \, \1 .
\end{align*}
Together with \eqref{eq:Mourre-19}, this shows that
\begin{align}
- P_{\Omega_a\times\Omega_a} &\ge  - 3 (C_1 +C_2) |g| \1 - \1_\Delta^\perp . \label{eq:Mourre-20}
\end{align}

To bound the last term in the right hand side of \eqref{eq:mourre-5}, it suffices to use the relative bound in Corollary~\ref{cor-rel-bound} and the fact that Hypothesis \ref{hyp:2} holds (and hence also Hypothesis \ref{hyp:1}), to obtain that the operators $H_I(- i a_l F)$ ($l=1,2)$ are norm relatively bounded with respect to $H_0$ with relative bounds depending on $K(G)$ and $K(- i a_l G)$. Therefore, there exists $C_3$ depending on $K(G)$ and $K(- i a_l G)$ such that
\begin{align}
 g  \big ( H_I (-i \ael F) + H_I(-i\ap F) \big ) &\ge - C_3 |g| \langle H \rangle \notag \\
 & = - C_3 |g| \langle H \rangle \1_\Delta( H - E ) - C_3 |g| \langle H \rangle \1_\Delta^\perp( H - E ) \notag \\
& \ge - C_4 |g| \1_\Delta( H - E ) - C_3 |g| \langle H \rangle \1_\Delta^\perp( H - E ) \notag \\
& \ge - C_4 |g| \1_\Delta( H - E ) - C_3 |g| \langle H \rangle \1_\Delta^\perp( H - E )  \notag \\
& \ge - C_4 |g| \1 - C_5 |g| \langle H \rangle \1_\Delta^\perp( H - E ) , \label{eq:mourre-17}
\end{align}
for some constants $C_4, C_5 \in \mathbb{R}$.

The estimates \eqref{eq:mourre-5}, \eqref{eq:Mourre-20} and \eqref{eq:mourre-17}  yield \eqref{eq:main-mourre}, which concludes the proof.
\end{proof}

\subsection{Proofs of the main theorems}\label{subsection:proofs}

\begin{proof}[Proof of Theorem \ref{thm:loc-spec}]
As above, we use the notation $E = \inf \mathrm{spec}(H)$. The proof of Theorem \ref{thm:loc-spec} is divided into two main steps.

\vspace{0,2cm}

\noindent \textbf{Step 1.} Let $0 < \delta < \me$. There exists $g_\delta > 0$ such that, for all $0 \le |g | \le g_\delta$,
\begin{equation*}
\inf \big ( \mathrm{spec}( H ) \setminus \{ E \} ) \ge \delta .
\end{equation*}
To prove this, we use the min-max principle.
Let $\mu_2$ denote the second point above $E$ in the spectrum of $H$. The min-max principle implies that
\begin{align*}
\mu_2 \ge \underset{ \begin{small} \begin{array}{c} \psi \in \mathfrak{D}( H ) , \| \psi \|=1 , \\ \psi \in [ \Omega_D \otimes \Omega_s ]^\perp \end{array} \end{small} }{ \inf } \langle  \psi , H \psi \rangle = \underset{ \begin{small} \begin{array}{c} \psi \in \mathfrak{D}( H ) , \| \psi \|=1 , \\ \psi \in [ \Omega_D \otimes \Omega_s ]^\perp \end{array} \end{small} }{ \inf } \big (  \langle  \psi , H_0 \psi \rangle + g \langle  \psi , H_I \psi \rangle \big ),
\end{align*}
where $[ \Omega_D \otimes \Omega_s ]^\perp$ denotes the orthogonal complement of the subspace spanned by $\Omega_D \otimes \Omega_s$ in the total Hilbert space $\mathcal{H}$. Since $H_I$ is relatively bounded with respect to $H_0$, there exists a positive constant $\mathrm{C}$ such that $\langle \psi , H_I \psi \rangle \ge - \mathrm{C} \langle \psi , H_0 \psi \rangle$, and therefore
\begin{align*}
\mu_2 \ge \underset{ \begin{small} \begin{array}{c} \psi \in \mathfrak{D}( H ) , \| \psi \|=1 , \\ \psi \in [ \Omega_D \otimes \Omega_s ]^\perp \end{array} \end{small} }{ \inf } ( 1 - \mathrm{C} |g| )  \langle  \psi , H_0 \psi \rangle \ge ( 1 - \mathrm{C} |g| ) \me ,
\end{align*}
the last inequality being a consequence of \eqref{eq:spec_H0}. This proves Step 1.

\vspace{0,2cm}

\noindent \textbf{Step 2.} Let $0 < \delta < \me$. There exists $g_\delta > 0$ such that, for all $0 \le |g | \le g_\delta$,
\begin{equation*}
\mathrm{spec}( H ) \cap [ \delta , \me + E ) = \emptyset.
\end{equation*}
Observe that $E<0$ satisfies $E \ge - \mathrm{C} |g|$ with $\mathrm{C}$ a positive constant, as follows from standard perturbation theory (see \eqref{eq:mourre-14}), and therefore, for $g_\delta$ small enough and $|g|Ê\le g_\delta$, we have that $\delta < \me + E$. By Theorem \ref{thm:ess}, we know that $\inf \mathrm{spec}_{ \mathrm{ess} }( H ) = \me + E$. Thus we only have to show that $H$ do not have discrete eigenvalue in the interval $[ \delta , \me + E )$: This is a simple, usual consequence of the virial theorem (see Theorem \ref{prop:virial}) combined with the Mourre estimate of Proposition \ref{prop:Mourre}.
\end{proof}

We introduce the notation $\langle A \rangle = (1+ A^* A )^{1/2} = (1 + |A|^2 )^{1/2}$ for any closed operator $A$. As mentioned before, Theorem \ref{thm:abs-cont} is a consequence of the following stronger result, which itself follows from Propositions \ref{prop:semigroup}, \ref{thm:interpo}, \ref{prop:Mourre}, and the abstract results of Appendix \ref{appendix:Mourre}.
\begin{theo}[Limiting absorption principle]\label{thm:LAP}
Assume that Hypothesis \ref{hyp:2} holds with $\varepsilon > 0$ in Hypothesis \ref{hyp:2}(ii). For all $\delta > 0$, there exists $g_\delta > 0$ such that, for all $| g | \le g_\delta$ and $1/2<s\le 1$,
  \begin{equation*}
    \sup_{z\in \tilde \Delta } \| \langle A \rangle^{-s} (H-z)^{-1} \langle A \rangle^{-s}\|<\infty ,
  \end{equation*}
with $\Delta := [ \inf \mathrm{spec}( H ) + \me ,\, \inf \mathrm{spec}( H ) + \mz - \delta ]$ and
$
\tilde \Delta := \{ z \in \mathbb{C} , \mathrm{Re} \, z\in \Delta , 0<|\mathrm{Im} \, z|\leq 1\} , .
$
Moreover, the map $z\mapsto \langle A \rangle^{-s} (H-z)^{-1} \langle A \rangle^{-s} \in \mathcal B (\mathcal H)$ is uniformly H\"older continuous of order $s-1/2$ on $\tilde \Delta$ and the limits
\begin{equation*}
\langle A \rangle^{-s} ( H - \lambda - i 0^{\pm})^{-1} \langle A \rangle^{-s} := \lim_{\varepsilon \to 0^{\pm} } \langle A \rangle^{-s} (H-\lambda-i\varepsilon)^{-1} \langle A \rangle^{-s},
\end{equation*}
exist in the norm topology of $\mathcal{B}( \mathcal{H} )$, uniformly in $\lambda \in \Delta$. Finally, the map $\lambda\mapsto  \langle A \rangle^{-s} (H-\lambda- i0^\pm)^{-1} \langle A \rangle^{-s} \in \mathcal B (\mathcal H)$ is uniformly H\"older continuous of order $s-1/2$ on $\Delta$ and, for any $1/2 < s \le 1$, $H$ satisfies the local decay property
\begin{equation}
\big \| \langle A \rangle^{-s} e^{ - i t H } \1_\Delta ( H ) \langle A \rangle^{-s} \big \| \lesssim \langle t \rangle^{ - s + \frac12 }  , \label{eq:local-decay}
\end{equation}
for all $t \in \mathbb{R}$.
\end{theo}
\begin{proof}
By Propositions \ref{prop:semigroup}, \ref{thm:interpo} and \ref{prop:Mourre}, we see that Hypotheses \ref{hyp:M1}, \ref{hyp:M3} and \ref{hyp:M5} of Appendix \ref{appendix:Mourre} are satisfied, the open interval $I$ of Hypothesis \ref{hyp:M5} being chosen, for instance, as $I = ( \inf \mathrm{spec}( H ) + \me - \delta ,\, \inf \mathrm{spec}( H ) + \mz - \delta /2 )$. Therefore we can apply Theorem \ref{thm:abstractLAP} with $J = \Delta$, which proves Theorem \ref{thm:LAP}.
\end{proof}

%%%%%%%%%%%%%%%%%%%%%%%%%%%%%%%%%%%%%%%%%%%%%%%%%%%%%%%%
%%%%%%%%%%%%%%%%%%%%  APPENDIX  %%%%%%%%%%%%%%%%%%%%%%%%
%%%%%%%%%%%%%%%%%%%%%%%%%%%%%%%%%%%%%%%%%%%%%%%%%%%%%%%%
\appendix

\section{Generalized eigenfunctions of the free Dirac operator}\label{appendixA}

In this section we describe the properties of the generalized eigenfunctions of the Dirac operator $H_D$ introduced in subsection~\ref{S2.1.1}. More details can be found in \cite[section 9.9, (44), (45), (63)]{Greiner}.

Recall that the generalized eigenfunctions of $H_D$ are labeled by the angular momentum
quantum numbers
 $$
    j \in\{ \frac{1}{2}, \frac{3}{2}, \frac{5}{2},
    \ldots\},\quad\ m_j \in\{ -j, -j+1, \ldots, j-1, j\},
 $$
 and by the quantum numbers
 $$
    \kappa_j \in \{\pm ( j+ \frac{1}{2})\} \, .
 $$
We define, for $\gamma_j :=|\kappa_j|$,
\begin{equation}\label{eq:def-g}
\begin{split}
  & g_{\kappa_j,\pm}(p,r) = \frac{C_1^\pm}{| \w(p) |^\frac12} \frac{(2 p r )^{\gamma_j}}{r} \frac{1}{2\sqrt{\pi}} \frac{\Gamma(\gamma_j)}{\Gamma(2\gamma_j+1)} \\
  & \times \left\{ \mathrm{e}^{-i p r} \mathrm{e}^{i\eta_j} \gamma_j F(\gamma_j+1,\, 2\gamma_j +1,\, 2 i p r)
  + \mathrm{e}^{i p r} \mathrm{e}^{-i\eta_j} \gamma_j F(\gamma_j +1,\, 2\gamma_j +1,\, -2 i p r)
  \right\}
\end{split}
\end{equation}
with
$C_1^+ = \sqrt{ \w(p)  + \me}$ when we consider a positive energy $ \w(p) >\me$ and $C_1^- = \sqrt{ \w(p)  -\me}$ when we consider a negative energy  $- \w(p) <-\me$.

We also define
\begin{equation}\label{eq:def-f}
\begin{split}
 & f_{\kappa_j,\pm}(p, r) = \frac{i C_2^\pm}{| \w(p) |^\frac12} \frac{1}{2\sqrt{\pi}} \frac{(2 p r)^{\gamma_j}}{r} \frac{\Gamma(\gamma_j)}{\Gamma(2\gamma_j+1)}\\
 & \times\left\{\mathrm{e}^{-i p r}\mathrm{e}^{i\eta_j} \gamma_j F(\gamma_j +1,\, 2\gamma_j +1,\, 2 i p r) - \mathrm{e}^{i p r}\mathrm{e}^{-i\eta_j} \gamma_j F(\gamma_j +1,\, 2\gamma_j +1,\, -2 i p r)  \right\}
\end{split}
\end{equation}
with
$C_2^+ = \sqrt{ \w(p)  - \me}$, for energies $ \w(p) >\me$ and $C_2^- = - \sqrt{\w(p) +\me}$ for energies $- \w(p) <-\me$.

The functions $F$ that occur in \eqref{eq:def-g} and \eqref{eq:def-f} are the confluent hypergeometric functions. Their integral representations for $\gamma_j>1/2$ are
\begin{equation}
\begin{split}
  F(\gamma_j+1,\, 2\gamma_j+1,\, \pm 2i pr)
 =  \frac{\Gamma(2\gamma_j +1)}{\Gamma(\gamma_j+1) \Gamma(\gamma_j)} \int_0^1 \mathrm{e}^{\pm 2 i p r u}
 u^{\gamma_j} (1-u)^{\gamma_j} \d u\, .
\end{split}
\end{equation}

The generalized eigenfunctions
 $$
   \pjpm (p, x) = \psi_{\pm,\gamma} (p, x) = \psipm (\xi, x)\, ,
 $$
 where $+$ refers to positive energies $ \w(p)  > \me$ and $-$ refers to negative energies $- \w(p)  < -\me$, fulfill
 $$
  H_D\
  \psipm((p,\gamma), x) = \pm\,  \w(p)\  \psipm((p,\gamma), x) \, ,
 $$
 and are defined by
\begin{equation}
\begin{split}
 \pjpm (p, x) :=
 \begin{pmatrix}
  i g_{\kappa_j,\pm}(p,r) \Phi^{(1)}_{m_j,\kappa_j} (\theta,\varphi) \\
  - f_{\kappa_j,\pm}(p,r) \Phi^{(2)}_{(m_j,\kappa_j)} (\theta,\varphi)
 \end{pmatrix}
\end{split}
\end{equation}
where the spinors $\Phi^{(1)}_{m_j,\kappa_j)}$ and $\Phi^{(2)}_{m_j,\kappa_j)}$ are orthogonal and defined by
\begin{equation}
\Phi^{(1)}_{(m_j,(j+\frac12))}(\theta,\varphi) :=
\begin{pmatrix}
  \sqrt{\frac{j-m_j+1}{2j+2}} Y_{j+\frac12, m_j-\frac12} (\theta,\varphi) \\
  -\sqrt{\frac{j+m_j+1}{2j+2}} Y_{j+\frac12, m_j-\frac12} (\theta,\varphi)
\end{pmatrix}
\end{equation}
\begin{equation}
\Phi^{(2)}_{(m_j,(j+\frac12))}(\theta,\varphi) :=
\begin{pmatrix}
  \sqrt{\frac{j+m_j}{2j}} Y_{j-\frac12, m_j-\frac12} (\theta,\varphi) \\
  \sqrt{\frac{j-m_j}{2j}} Y_{j-\frac12, m_j+\frac12} (\theta,\varphi)
\end{pmatrix}
\end{equation}
and
\begin{equation}
  \begin{split}
    \Phi^{(1)}_{(m_j,-(j+\frac12))}(\theta,\varphi) =&\Phi^{(2)}_{(m_j,(j+\frac12))}(\theta,\varphi)\\
    \Phi^{(2)}_{(m_j,-(j+\frac12))}(\theta,\varphi) =& -\Phi^{(1)}_{(m_j,(j+\frac12))}(\theta,\varphi).
  \end{split}
\end{equation}

It follows from \eqref{intro2} that
\begin{equation*}
  \begin{split}
 \widetilde{\psi}_{-,(j ,m_j,\kappa_j)} (p, x) :=
 \begin{pmatrix}
  i g_{-\kappa_j,-}(p,r) \Phi^{(1)}_{(-m_j,-\kappa_j)} (\theta,\varphi) \\
  - f_{-\kappa_j,-}(p,r) \Phi^{(2)}_{(-m_j,-\kappa_j)} (\theta,\varphi)
 \end{pmatrix}
\end{split}
\end{equation*}

For positive energies $ \w(p)  > \me$, we have the following estimates for the functions $g_{\kappa_j,\pm}$ and $f_{\kappa_j,\pm}$,
\begin{equation}
\begin{split}
 | g_{j+\frac12, +}(p,r)| & \leq \left(\frac{ \w(p) + \me }{ \w(p) }\right)^\frac12
 \frac{p}{\sqrt{\pi}} (2 p r)^{\gamma_j} \frac{1}{\Gamma(\gamma_j)} \, ,\\
 | f_{j+\frac12, +}(p,r)| & \leq \left(\frac{ \w(p) - \me }{ \w(p) }\right)^\frac12
 \frac{2p}{\sqrt{\pi}} (2 p r)^{\gamma_j -1} \frac{1}{\Gamma(\gamma_j)} \, ,\\
 | g_{-(j+\frac12), +}(p,r)| & \leq \left(\frac{ \w(p) + \me }{ \w(p) }\right)^\frac12
 \frac{2p}{\sqrt{\pi}} (2 p r)^{\gamma_j-1} \frac{1}{\Gamma(\gamma_j)} \, ,\\
 | f_{-(j+\frac12), +}(p,r)| & \leq \left(\frac{ \w(p) - \me }{ \w(p) }\right)^\frac12
 \frac{p}{\sqrt{\pi}} (2 p r)^{\gamma_j} \frac{1}{\Gamma(\gamma_j)} \, ,
\end{split}
\end{equation}
and for negative energies $- \w(p) <- \me$ , we have
\begin{equation}
\begin{split}
 | g_{j+\frac12, -}(p,r)| & \leq \left(\frac{ \w(p) - \me }{ \w(p) }\right)^\frac12
 \frac{p}{\sqrt{\pi}} (2 p r)^{\gamma_j} \frac{1}{\Gamma(\gamma_j)} \, ,\\
 | f_{j+\frac12, -}(p,r)| & \leq \left(\frac{ \w(p) + \me }{ \w(p) }\right)^\frac12
 \frac{2p}{\sqrt{\pi}} (2 p r)^{\gamma_j -1} \frac{1}{\Gamma(\gamma_j)} \, ,\\
 | g_{-(j+\frac12), -}(p,r)| & \leq \left(\frac{ \w(p) - \me }{ \w(p) }\right)^\frac12
 \frac{2p}{\sqrt{\pi}} (2 p r)^{\gamma_j-1} \frac{1}{\Gamma(\gamma_j)} \, ,\\
 | f_{-(j+\frac12), -}(p,r)| & \leq \left(\frac{ \w(p) + \me }{ \w(p) }\right)^\frac12 \frac{p}{\sqrt{\pi}} (2 p r)^{\gamma_j} \frac{1}{\Gamma(\gamma_j)}
 \, .\\
\end{split}
\end{equation}

We also can bound the first and second derivatives. Below, we give such bounds for $|p|\leq 1$. For $p$ larger than one, the functions are locally in $L^q$ for any value of $q$.

There exists a constant $C$ such that for $|p|\leq 1$, and for positive energies $\omega(p) > \me$ we have
\begin{equation}\label{eq:deriv-estimate-1}
\begin{split}
  \left| \frac{\partial }{\partial p}\, g_{j+\frac12, +}(p,r) \right|&
 \leq
 \frac{C}{\Gamma(\gamma_j)}\left[ (2pr)^{\gamma_j} + p r(\gamma_j-1)  (2pr)^{\gamma_j -1} + p r(2pr)^{\gamma_j-1} \right] \, ,
 \\
 \left|\frac{\partial }{\partial p}\, f_{j+\frac12, +}(p,r) \right| &
 \leq
 \frac{C}{\Gamma(\gamma_j)}\left[ p(2pr)^{\gamma_j-1} + p^2 r(\gamma_j-1)  (2pr)^{\gamma_j -2} + p^2 r (2pr)^{\gamma_j} \right] \, ,
 \\
 \left|\frac{\partial }{\partial p}\, g_{-(j+\frac12), +}(p,r) \right| &
 \leq
 \frac{C}{\Gamma(\gamma_j)}\left[ (2pr)^{\gamma_j-1} + p r (\gamma_j-1)  (2pr)^{\gamma_j -2} + p r (2pr)^{\gamma_j} \right]  \, ,
 \\
 \left|\frac{\partial }{\partial p}\, f_{-(j+\frac12), +}(p,r) \right| &
 \leq
 \frac{C}{\Gamma(\gamma_j)}\left[ p(2pr)^{\gamma_j} + p^2 r (\gamma_j-1) (2pr)^{\gamma_j -1} + p^2 r (2pr)^{\gamma_j-1} \right] \, ,
\end{split}
\end{equation}
and for $|p|\leq 1$ and negative energies $-\omega(p) < -\me$, we have
\begin{equation}\label{eq:deriv-estimate-2}
\begin{split}
  \left| \frac{\partial }{\partial p}\, g_{j+\frac12, -}(p,r) \right| &
 \leq
 \frac{C}{\Gamma(\gamma_j)}\left[ p(2pr)^{\gamma_j} + p^2 r (\gamma_j-1) (2pr)^{\gamma_j -1} + p^2 r (2pr)^{\gamma_j-1} \right]\, ,
 \\
 \left|\frac{\partial }{\partial p}\, f_{j+\frac12, -}(p,r) \right| &
 \leq
 \frac{C}{\Gamma(\gamma_j)}\left[ (2pr)^{\gamma_j-1} + p r (\gamma_j-1)  (2pr)^{\gamma_j -2} + p r (2pr)^{\gamma_j} \right]\, ,
 \\
 \left|\frac{\partial }{\partial p}\, g_{-(j+\frac12), -}(p,r) \right| &
 \leq
 \frac{C}{\Gamma(\gamma_j)}\left[ p(2pr)^{\gamma_j-1} + p^2 r(\gamma_j-1)  (2pr)^{\gamma_j -2} + p^2 r (2pr)^{\gamma_j} \right]\, ,
 \\
 \left|\frac{\partial }{\partial p}\, f_{-(j+\frac12), -}(p,r) \right| &
 \leq
 \frac{C}{\Gamma(\gamma_j)}\left[ (2pr)^{\gamma_j} + p r(\gamma_j-1)  (2pr)^{\gamma_j -1} + p r(2pr)^{\gamma_j-1} \right]\, .
\end{split}
\end{equation}

The estimates \eqref{eq:deriv-estimate-1} and \eqref{eq:deriv-estimate-2} yield, for $a$ being the operator defined by \eqref{eq:defa}, and for positive energies $\omega(p) > \me$,
\begin{equation}
\begin{split}
 | a \, g_{j+\frac12, +}(p,r) |
 \leq \frac{C}{\Gamma(\gamma_j)}&
 \Big[ \frac{\omega(p)}{p} \left((2pr)^{\gamma_j} + p r(\gamma_j-1)  (2pr)^{\gamma_j -1} + p r(2pr)^{\gamma_j-1} \right) \\
 & +
 \omega(p) (1+\frac{1}{p^2}) p (2pr)^{\gamma_j}) \Big] \, ,
 \\
 |a\, f_{j+\frac12, +}(p,r) |
 \leq \frac{C}{\Gamma(\gamma_j)}&
 \Big[ \frac{\omega(p)}{p} \left(p(2pr)^{\gamma_j-1} + p^2 r(\gamma_j-1)  (2pr)^{\gamma_j -2} + p^2 r(2pr)^{\gamma_j} \right) \\
 & +
 \omega(p) (1+\frac{1}{p^2}) p^2 (2pr)^{\gamma_j}) \Big] \, ,
 \\
 | a \, g_{-(j+\frac12), +}(p,r) |
 \leq \frac{C}{\Gamma(\gamma_j)}&
 \Big[ \frac{\omega(p)}{p} \left( (2pr)^{\gamma_j-1} + pr (\gamma_j-1)  (2pr)^{\gamma_j -2} + p r(2pr)^{\gamma_j} \right) \\
 & +
 \omega(p) (1+\frac{1}{p^2}) p (2pr)^{\gamma_j-1}) \Big] \, ,
 \\
 | a \, f_{-(j+\frac12), +}(p,r) |
 \leq \frac{C}{\Gamma(\gamma_j)}&
 \Big[ \frac{\omega(p)}{p} \left(p(2pr)^{\gamma_j} + p^2 (\gamma_j-1)  (2pr)^{\gamma_j -1} + p^2 r(2pr)^{\gamma_j-1} \right) \\
 & +
 \omega(p) (1+\frac{1}{p^2}) ¨p^2 (2pr)^{\gamma_j}) \Big] \, ,
 \\
\end{split}
\end{equation}
And for negatives energies $-\omega(p) < -\me$, we get the same estimates for $| a \, g_{j+\frac12, -}(p,r) |$, $|a\, f_{j+\frac12, -}(p,r) | $, $|a \, g_{-(j+\frac12), -}(p,r) |$ and $| a \, f_{-(j+\frac12), -}(p,r) |$, respectively for
$| a \, f_{-(j+\frac12), +}(p,r) |$, $|a \, g_{-(j+\frac12), +}(p,r) |$, $|a\, f_{j+\frac12, +}(p,r) | $ and $| a \, g_{j+\frac12, +}(p,r) |$.

Estimates for the second derivatives are given for $(p,r)$ near $(0,0)$ by
\begin{equation}
\begin{split}\label{eq:deriv2nd-estimate-1}
 \left| \frac{\partial ^2}{\partial p^2}\, g_{j+\frac12, +}(p,r) \right|
 & \le \frac{ C \gamma_j^2 }{ \Gamma( \gamma_j ) } p^{\gamma_j-1} r^{\gamma_j},
\end{split}
\end{equation}
\begin{equation}
\begin{split}
 \left| \frac{\partial ^2}{\partial p^2}\, f_{j+\frac12, +}(p,r) \right|
  & \le \frac{ C \gamma_j^2 }{ \Gamma( \gamma_j ) } p^{\gamma_j-1} r^{\gamma_j-1},
\end{split}
\end{equation}
\begin{equation}
\begin{split}
 \left| \frac{\partial ^2}{\partial p^2}\, g_{-(j+\frac12), +}(p,r) \right|
  & \le \frac{ C \gamma_j^2 }{ \Gamma( \gamma_j ) } p^{\gamma_j-2} r^{\gamma_j-1},
\end{split}
\end{equation}
\begin{equation}
\begin{split}\label{eq:deriv2nd-estimate-4}
 \left| \frac{\partial ^2}{\partial p^2}\, f_{-(j+\frac12), +}(p,r) \right|
  & \le \frac{ C \gamma_j^2 }{ \Gamma( \gamma_j ) } p^{\gamma_j-1} r^{\gamma_j-1},
\end{split}
\end{equation}
and the same estimates for negatives energies hold respectively for $\left| \frac{\partial ^2}{\partial p^2}\, f_{-(j+\frac12), -}(p,r) \right|$,
$\left| \frac{\partial ^2}{\partial p^2}\, g_{-(j+\frac12), -}(p,r) \right|$, $\left| \frac{\partial ^2}{\partial p^2}\, f_{j+\frac12, -}(p,r) \right|$ and $\left| \frac{\partial ^2}{\partial p^2}\, g_{j+\frac12, -}(p,r) \right|$\, .

\section{Mourre theory: abstract framework}\label{appendix:Mourre}

In this section, we recall some abstract results from Mourre's theory that were used in Section \ref{sec:LAP}. We work with an extension of the original Mourre theory \cite{Mo} that allows, in particular, the so-called \emph{conjugate operator} to be maximal symmetric (not necessarily self-adjoint). Such an extension was considered in \cite{HuSp} and further refined in \cite{GGM1,GGM2} (see also \cite{FMS,Go}). Here we mainly follow the presentation of \cite{FMS}.

Let $\mathcal{H}$ be a complex separable Hilbert space. Consider a self-adjoint operator $H$ on $\mathcal{H}$ and a symmetric operator $H'$ on $\mathcal{H}$ such that $\D(H) \subset \D( H' )$. Let
\begin{equation*}
\mathcal{G} := \D( | H |^{\frac{1}{2}} ) ,
\end{equation*}
equipped with the norm
\begin{equation*}
\| \varphi \|^2_{ \mathcal{G} } :=  \big \| |H|^{\frac{1}{2}} \varphi \big \|^2 + \| \varphi \|^2.
\end{equation*}
We set
\begin{equation*}
\| \varphi \|^2_{ \mathcal{G}^* } :=  \big \| ( |H| + \1 )^{-\frac{1}{2}} \varphi \big \|^2.
\end{equation*}
The dual space $\mathcal{G}^*$ of $\mathcal{G}$ identifies with the completion of $\mathcal{H}$ with respect to the norm $\| \cdot \|_{\mathcal{G}^*}$, and the operator $H$ identifies with an element of $\mathcal{B}( \mathcal{G} ; \mathcal{G}^* )$, the set of bounded operators from $\mathcal{G}$ to $\mathcal{G}^*$.

Let $A$ be a closed and maximal symmetric operator on $\mathcal{H}$. In particular, the deficiency indices $n_{\mp}=\dim \mathrm{Ker} (A^*\pm i)$ of $A$ obey either $n_+=0$ or $n_-=0$. We suppose that $n_-=0$ so that $-A$  generates a
$\mathrm{C}_0$-semigroup of isometries $\{ W_t \}_{t \ge 0}$ (see  e.g. \cite[Theorem 10.4.4]{Da}). Recall  that a $\mathrm{C}_0$-semigroup on $[ 0 , \infty )$ is, by definition, a map $t \mapsto W_t \in \mathcal{B}( \mathcal{ H } )$ such that $W_0 = \1$, $W_tW_s=W_{t+s}$ for $t,s\ge0$, and $\mathrm{w}\text{-}\lim_{t\to 0^+} W_t = \1$, where $\mathcal{B}( \mathcal{H})$ denotes the set of bounded operators on $\mathcal{H}$ and $\mathrm{w}\text{-}\lim$ stands for weak limit. The fact that $-A$ is the generator of the $\mathrm{C}_0$-semigroup $\{ W_t \}_{t \ge 0}$ means that
\begin{align*}
& \mathfrak{D} (A) = \big \{ u \in \mathcal{H} ,\lim_{t \to 0^+} ( i t )^{-1}
( W_t u - u ) \text{ exists} \big \} , \\
&-iAu = \lim_{t \to 0^+} t^{-1}
( W_t u - u ).
\end{align*}
We make the following hypotheses.
\begin{hypo}\label{hyp:M1}
For all $t > 0$, $W_t$ and $W_t^*$ preserve $\mathcal{G}$ and, for all $\varphi \in \mathcal{G}$,
\begin{equation*}
\sup_{0<t<1} \| W_t \varphi \|_\mathcal{G} <  \infty, \quad \sup_{0<t<1} \| W_t^* \varphi \|_\mathcal{G} < \infty.
\end{equation*}
In particular, $t \mapsto W_t |_\mathcal{G} \in \mathcal{B}( \mathcal{G} )$ is a $\mathrm{C}_0$-semigroup, and the extension of $W_t$ to $\mathcal{G}^*$ (which will be denoted by the same symbol) defines a $\mathrm{C}_0$-semigroup on $\mathcal{B}( \mathcal{G}^* )$ (see \cite[Remark 1.4.1)]{FMS}. Their generators are denoted by $A_\mathcal{G}$ and $A_{\mathcal{G}^*}$, respectively.
\end{hypo}
\begin{hypo}\label{hyp:M2}
The operator $H \in \mathcal{B}( \mathcal{G} ; \mathcal{G}^* )$ is of class $\mathrm{C}^1( A_{ \mathcal{G} } ; A_{\mathcal{G}^*} )$, meaning that there exists a positive constant $\mathrm{C}$ such that, for all $0 \le t \le 1$,
\begin{equation*}
\| W_{t} H - H W_{t} \|_{ \mathcal{B} ( \mathcal{G} ; \mathcal{G}^* ) } \le \mathrm{C} t .
\end{equation*}
Moreover, for all $\varphi \in \D( H )$,
\begin{equation*}
\lim_{ t \to 0^+ } \big ( \langle \varphi , W_t H \varphi \rangle - \langle H \varphi , W_t \varphi \rangle \big ) = \langle \varphi , H' \varphi \rangle.
\end{equation*}
\end{hypo}
\begin{prop}\label{prop:B3}
Suppose that Hypothesis \ref{hyp:M1} holds and that the sesquilinear form $[H,iA]$ defined on $\mathfrak{D}(A) \cap \mathcal{G}$ by
\begin{equation*}
\langle u , [ H , i A ] v \rangle := i \langle u , H A v \rangle - i \langle A^* u , H v \rangle ,
\end{equation*}
extends to a bounded quadratic form on $\mathcal{G}$. Then $H$ is of class $\mathrm{C}^1( A_{ \mathcal{G} } ; A_{\mathcal{G}^*} )$ in the sense Hypothesis \ref{hyp:M2}.
\end{prop}
Under Hypotheses \ref{hyp:M1} and \ref{hyp:M2}, we have the following version of the virial theorem.
\begin{theo}[Virial Theorem]\label{prop:virial}
Assume Hypotheses \ref{hyp:M1} and \ref{hyp:M2}. For any eigenstate $\varphi$ of $H$, we have that
\begin{equation*}
\langle \varphi , H' \varphi \rangle = 0.
\end{equation*}
\end{theo}
The limiting absorption principle stated in Theorem \ref{thm:abstractLAP} below requires some more regularity of $H$ with respect to $A$:
\begin{hypo}\label{hyp:M3}
The operator $ H \in \mathcal{B}( \mathcal{G} ; \mathcal{G}^* )$ is of class $\mathrm{C}^{1,1}( A_{ \mathcal{G} } ; A_{\mathcal{G}^*} )$, i.e.
\begin{equation*}
\int_0^1 \big \| [ W_t  ,[ W_t , H ] ] \big \|_{ \mathcal{B} ( \mathcal{G} ; \mathcal{G}^* ) } \frac{ \d t }{ t^2 } < \infty.
\end{equation*}
\end{hypo}
We recall that $\langle A \rangle = (1+ A^* A )^{1/2} = (1 + |A|^2 )^{1/2}$ for any closed operator $A$.  Our last hypothesis is a version of a strict Mourre estimate.
\begin{hypo}\label{hyp:M5}
There exist an open interval $I \subset \mathbb{R}$ and constants $\mathrm{c}_0>0$, $\mathrm{C} \in \mathbb{R}$, such that, in the sense of quadratic forms on $\mathcal{D}( H )$,
\begin{equation}\label{eq:Mourre_estimate}
 H' \ge \mathrm{c}_0 \1 - \mathrm{C} \1_I^\perp(H) \langle H \rangle ,
\end{equation}
where $\1_I^\perp(H) := \1 - \1_I(H)$.
\end{hypo}
The following theorem shows that a limiting absorption principle holds for $H$ in any compact interval where a Mourre estimate is satisfied in the sense of Hypothesis \ref{hyp:M5}. The proof of Theorem \ref{thm:abstractLAP} can be found in \cite{GGM1} (see also \cite{HuSp} for a similar result under slightly stronger assumptions).
\begin{theo}[Limiting absorption principle]\label{thm:abstractLAP}
Assume that Hypotheses \ref{hyp:M1}, \ref{hyp:M3} and \ref{hyp:M5} hold. Let $J \subset I$ be a compact interval, where $I$ is given by Hypothesis \ref{hyp:M5}, and let
\begin{equation*}
\tilde J = \{ z \in \mathbb{C} , \mathrm{Re} \, z\in J , 0<|\mathrm{Im} \, z|\leq 1\}.
\end{equation*}
For any $1/2<s\le 1$, we have that
  \begin{equation*}
    \sup_{z\in \tilde J}\| \langle A \rangle^{-s} (H-z)^{-1} \langle A \rangle^{-s}\|<\infty ,
  \end{equation*}
and the map $z\mapsto \langle A \rangle^{-s} (H-z)^{-1} \langle A \rangle^{-s} \in \mathcal B (\mathcal H)$ is uniformly H\"older continuous of order $s-1/2$ on $\tilde J$. In particular, the limits
  \begin{equation*}
  \langle A \rangle^{-s} ( H - \lambda - i 0^{\pm})^{-1} \langle A \rangle^{-s} := \lim_{\varepsilon \to 0^{\pm} } \langle A \rangle^{-s} (H-\lambda-i\varepsilon)^{-1} \langle A \rangle^{-s},
  \end{equation*}
exist in the norm topology of $\mathcal{B}( \mathcal{H} )$, uniformly in $\lambda \in J$. This implies that the spectrum of $H$ in $J$ is purely absolutely continuous. Moreover, the map $\lambda\mapsto  \langle A \rangle^{-s} (H-\lambda- i0^\pm)^{-1} \langle A \rangle^{-s} \in \mathcal B (\mathcal H)$ is uniformly H\"older continuous of order $s-1/2$ on $J$.
\end{theo}
\begin{rema}
\begin{itemize}
\item[1)] Theorem \ref{thm:abstractLAP} is established in \cite{GGM1} in the more general context of \emph{singular Mourre theory}. More precisely, as shown in \cite{GGM1}, the assumption that the commutator $H'$ is relatively bounded with respect to $H$ can be relaxed. This is of fundamental importance for the application to massless quantized fields considered in \cite{GGM2}, but is not needed for the model studied in the present paper. Therefore, we content ourselves with the simpler setting of \emph{regular Mourre theory} (i.e. we suppose that $H'$ is $H$-bounded).
\item[2)] The results in \cite{GGM1} are formulated under a stronger assumption than Hypothesis \ref{hyp:M3}, namely that $H \in \mathrm{C}^{2}( A_{ \mathcal{G} } ; A_{\mathcal{G}^*} )$. Nevertheless, as mentioned in \cite{GGM1}, one can verify that Hypothesis \ref{hyp:M3} is sufficient for Theorem \ref{thm:abstractLAP} to hold.
\item[3)] By Fourier transform, Theorem \ref{thm:abstractLAP} implies the local decay property
\begin{equation*}
\big \| \langle A \rangle^{-s} e^{ - i t H } \chi ( H ) \langle A \rangle^{-s} \big \| = \mathcal{O}( \langle t \rangle^{ - s + \frac12 } ) ,
\end{equation*}
for any $\chi \in \mathrm{C}_0^\infty( I ; \mathbb{R} )$ and $1/2 < s \le 1$.
\end{itemize}
\end{rema}

\section{Creation and annihilation operators in Fermi-Fock space}\label{appendixC}

Let ${\mathfrak G} $ be any separable Hilbert space. Let $\otimes_a^n
{\mathfrak G}$ denotes the antisymmetric $n$-th tensor power of ${\mathfrak G}$,
appropriate to Fermi-Dirac statistics. We define the Fermi-Fock space over
${\mathfrak G}$, denoted by ${\mathfrak F}_a({\mathfrak G})$, to be the direct
sum
 $$
  {\mathfrak F}_a({\mathfrak G}) = {
  \displaystyle\mathop \bigoplus_{n=0}^\infty} \, \otimes_a^n {\mathfrak G} ,
 $$
where, by definition, we have set $\otimes_a^0{\mathfrak G} := \C$. We shall denote by $\Omega_a$ the
vacuum vector in ${\mathfrak F}_a({\mathfrak G})$, i.e., the vector
 $(1,0,0,\cdots)$.

Let ${\mathfrak F}_{a}$ be the Fermi-Fock space over  $\mathfrak{
H}_{c}$,
$$
 {\mathfrak F}_{a} : = {\mathfrak F}_{a}(\gH_c)\, .
$$
The Fermi-Fock space for electrons and positrons, denoted
by ${\mathfrak F}_D$, is the following Hilbert space
\begin{eqnarray}
{\mathfrak F}_D := {\mathfrak F}_{a} \otimes
{\mathfrak F}_{a}\, .
\end{eqnarray}
We denote by $\Omega_D :=  \Omega_a \otimes \Omega_a$ the vacuum of electrons and positrons.
One has
 $$ {\mathfrak F}_D = \displaystyle \mathop \oplus_{r,s=0}^\infty
 {\mathfrak F}_a^{(r,s)} \, ,
 $$
where ${\mathfrak F}_a^{(r,s)} := (\otimes_a^r  {\mathfrak H}_{c}) \otimes
(\otimes_a^s  {\mathfrak H}_{c})$.

For every $\varphi \in \mathfrak{H}$ we define in ${\mathfrak F}_a(\mathfrak{H})$ the annihilation operator, denoted by $b(\varphi)$ as:
\begin{equation*}
b( \varphi ) \Omega = 0 ,
\end{equation*}
and, for any $n \in \mathbb{N}$,
\begin{eqnarray*}
&& b(\varphi)\, (A_{n+1}(\varphi_1 \otimes \ldots \otimes \varphi_{n+1}))\\
\noalign{\vskip 8pt} && \qquad = \frac{\sqrt{n+1}}{(n+1)!} \sum_\sigma
sgn(\sigma)\, (\varphi, \varphi_{\sigma(1)})\, \varphi_{\sigma(2)}\otimes
\ldots \otimes \varphi_{\sigma(n+1)}
\end{eqnarray*}
where $\varphi_i \in {\mathfrak H}$.  Note that the operator $b(\varphi)$ maps $\otimes_a^{n+1} {\mathfrak H}$ to $\otimes_a^n {\mathfrak H}$. It extends by linearity to a bounded operator on $\mathfrak{F}_a({\mathfrak H})$.

The creation  operator, denoted by $b^*(\varphi)$, is the adjoint of
$b(\varphi)$. The operators $b^*(\varphi)$ and $b(\varphi)$ satisfy $\|b(\varphi)\| = \|b^*(\varphi)\| =
\|\varphi\|$.

We now define the annihilation and creation operators in the Fermi-Fock space ${\mathfrak F}_D$ for electrons and positrons.

We first define the creation and annihilation operators for the electrons. For any $g\in \mathfrak{H}_{c}$, we define in ${\mathfrak F}_D = {\mathfrak F}_a \otimes \mathfrak{F}_a$ the
annihilation operator, denoted by $b_+(g)$, as
\begin{equation*}
b_+(g) := b(g) \otimes \1 .
\end{equation*}
Observe that $b_+(g)$ maps ${\mathfrak F}_a^{(r+1,s)}$ into ${\mathfrak F}_a^{(r,s)}$ as follows
\begin{equation}\nonumber
\begin{split}
 & b_+(g) \left(
 A_{r+1}(g_1 \otimes \ldots \otimes g_{r+1}) \otimes A_s(h_1 \otimes \ldots
 \otimes h_s) \right)\\
 & =   \left[ \, b(g) A_{r+1}
 (g_1 \otimes \ldots \otimes g_{r+1})
 \right] \otimes A_s(h_1 \otimes \ldots \otimes h_s)
\end{split}
\end{equation}
The creation operator $b_+^*(g) = b^*(g) \otimes \1$ is the adjoint of $b_+(g)$. The operators $b_+^*(g)$ and
$b_+(g)$ are bounded operators in ${\mathfrak F}_D$.

We set, for every $g \in \mathfrak H_{c}$,
\begin{eqnarray*}
b_{\gamma,+}(g)   &=& b_+  (P_\gamma^+ g)\\ \noalign{\vskip 8pt}
b_{\gamma,+}^*(g)&=& b_+^*(P_\gamma^+ g)
\end{eqnarray*}
where $P_\gamma^+$ is the projection of $\mathfrak H_{c}$ onto the
$\gamma$-th component defined according to \eqref{eq:def-unitary-gamma}.

We next define the creation and annihilation operators for the positrons. For every $h \in \mathfrak H_{c}$, we define in ${\mathfrak F}_D$ the
annihilation operator, denoted by $b_-(h)$, as
\begin{equation*}
b_-(h) := (-1)^{N_{\mathrm{e}}} \otimes b(h) ,
\end{equation*}
where $(-1)^{N_{\mathrm{e}}}$ denotes the bounded operator on $\mathfrak{F}_a$ defined by its restriction to $\otimes_a^r \mathfrak{h}_c$ as $(-1)^{N_{\mathrm{e}}} u = (-1)^r u$ for any $u \in \otimes_a^r \mathfrak{h}_c$.

In other words, $b_-(h)$ maps ${\mathfrak F}_a^{(r,s+1)}$ into ${\mathfrak F}_a^{(r,s)}$ as follows:
\begin{equation}\nonumber
\begin{split}
 & b_-(h) ( A_r(g_1 \otimes
 \ldots\otimes g_r) \otimes A_{s+1}(h_1 \otimes \ldots \otimes h_{s+1}))\\
 & =  A_r(g_1
 \otimes \ldots \otimes g_r) \otimes [(-1)^{r} b(h) A_{s+1}(h_1 \otimes \ldots
 \otimes h_{s+1})]
\end{split}
\end{equation}
The creation operator $b_-^*(h) = (-1)^{N_{\mathrm{e}}} \otimes b^*(h)$ is the adjoint of $b_-(h)$; $b_-^*(h)$ and
$b_-(h)$ are bounded operators in ${\mathfrak F}_D$.

As above, we set, for every $h \in \mathfrak H_{c}$,
\begin{eqnarray*}
b_{\gamma,-}(h)   &=& b_{-}  (P_\gamma^- h)\\ \noalign{\vskip 8pt}
b_{\gamma,-}^*(h) &=& b_-^*(P_\gamma^- h)
\end{eqnarray*}
where $P_\gamma^-$ is the projection of $\mathfrak H_{c}$ onto the
$\gamma$-th component.

A simple computation shows that the following anti-commutation relations hold
\begin{eqnarray*}
 && \{b_{\gamma,\pm}(g_1),  b_{\beta,\pm}^*(g_2)\} = \delta_{\gamma,\beta}(P_\gamma^\pm g_1, P_\gamma^\pm
 g_2)_{L^2(\R _+)}\, ,
\end{eqnarray*}
and
\begin{eqnarray*}
 \{b^{\sharp_1}_{\gamma,+}(g_1), b^{\sharp_2}_{\beta,-}(g_2)\} = 0\, ,
\end{eqnarray*}
where $g_1$, $g_2\in \mathfrak H_{c}$, and $\sharp_i$ ($i=1,2)$ stand either for $*$ or for no symbol.

As in \cite[chapter X]{RS}, we introduce
operator-valued distributions $b_{\gamma,\pm}(p)$ and $b_{\gamma,\pm}^*(p)$ that fulfills
\begin{equation}\nonumber
\begin{split}
 & b_{\gamma,\pm}(g)   = \int_{\R ^+}\
 b_{\gamma,\pm}(p)\,  \overline{(P_\gamma^\pm g)\, (p)} \, \d p\\
 & b_{\gamma,\pm}^*(g) = \int_{\R ^+}\ b_{\gamma,\pm}^*(p)\,
 (P_\gamma^\pm g)\, (p)\, \d p
\end{split}
\end{equation}
where $g\in \mathfrak H_{c}$.

We also define for $\xi=(p,\gamma)$,
$$
 b^\sharp_\pm(\xi) := b^\sharp_{\gamma,\pm}(p)\, .
$$
Note that with the notation of \eqref{eq:notation-int-dxi}, we have
\begin{equation}\nonumber
\begin{split}
 & b_{\pm}^\sharp(g)   = \int  b_{\pm}^\sharp(\xi)   \overline{ g (\xi)} \, \d \xi \, .
\end{split}
\end{equation}

We now give a representation of $b_{\gamma,\pm}(p)$ and $b_{\gamma,\pm}^*(p)$.
Recall that $\D_D$ denote the set of smooth vectors $\Phi \in {\mathfrak F}_D$
 for which $\Phi^{(r,s)}$ has a compact  support and $\Phi^{(r,s)} =
0$ for all but finitely
 many $(r,s)$.

For every $\xi_1=(p,\gamma)$, $b_{+}(\xi_1)$ maps ${\mathfrak F}_a^{(r+1,s)}
\cap \D_D$ into
 ${\mathfrak F}_a^{(r,s)} \cap \D_D$ and we have
\begin{eqnarray*}
 (b_{+}(\xi_1)\Phi)^{(r,s)}( p_1,\gamma_1,\ldots, p_r,\gamma_r ; p'_1,\gamma'_1,\ldots, p'_s,\gamma'_s)&&\\
  = \sqrt{r+1} \Phi^{(r+1,s)} (p,\gamma, p_1,\gamma_1,\ldots, p_r,\gamma_r ;p'_1,\gamma'_1,\ldots,
 p'_s,\gamma'_s)&&
\end{eqnarray*}
$b_{+}^*(\xi_1)$ is then given by:
\begin{eqnarray*}
 &&(b_{+}^*(\xi_1)\Phi)^{(r+1,s)}(
 p_1,\gamma_1,\ldots, p_{r+1},\gamma_{r+1} ; p'_1,\gamma'_1,\ldots,
 p'_s,\gamma'_s)=\\
 &&\frac{1}{\sqrt{r+1}}\sum_{i=1}^{r+1}(-1)^{i+1}
\delta_{\gamma_i \gamma}\delta(p-p_i)\\
 &&\Phi^{(r,s)}(p_1,\gamma_1,\ldots, \widehat{p_i,\gamma_i},\ldots,
 p_{r+1},\gamma_{r+1} ; p'_1,\gamma'_1,\ldots, p'_s,\gamma'_s)
\end{eqnarray*}
where $\widehat\cdot$ denotes that the i-th variable has to be
omitted.

Similarly, for $\xi_2=(p',\gamma')$, $b_{-}(\xi_2)$ maps ${\mathfrak F}_a^{(r,s+1)} \cap
\D_D$ into ${\mathfrak F}_a^{(r,s)} \cap \D_D$ such that
\begin{equation}\nonumber
\begin{split}
 \lefteqn{(b_{-}(\xi_2)\Phi)^{(r,s)}(
 p_1,\gamma_1,\ldots, p_r,\gamma_r ; p'_1,\gamma'_1,\ldots,
 p'_s,\gamma'_s)=} &\\
 \lefteqn{((-1)^{\Ne}\otimes b(\xi_2)\Phi)^{(r,s)}(
 p_1,\gamma_1,\ldots, p_r,\gamma_r ; p'_1,\gamma'_1,\ldots,
 p'_s,\gamma'_s)=} &\\
 & \sqrt{s+1}(-1)^{r} \Phi^{(r,s+1)}
 (p_1,\gamma_1,\ldots,
 p_r,\gamma_r ; p', \gamma', p'_1,\gamma'_1,\ldots, p'_s,\gamma'_s)
\end{split}
\end{equation}
$b_{-}^*(\xi_2)$ is then given by
\begin{eqnarray*}
 && (b_{-}^*(\xi_2)\Phi)^{(r,s+1)} (
 p_1,\gamma_1,\ldots, p_r,\gamma_r ; p'_1,\gamma'_1,\ldots, p'_{s+1},\gamma'_{s+1})=
 \\
 &&\frac{1}{\sqrt{s+1}}(-1)^{r} \sum_{i=1}^{s+1}(-1)^{i+1}
 \delta_{\gamma',\gamma'_i}\delta(p'-p'_i) \\
 &&\Phi^{(r,s)}(p_1,\gamma_1,\ldots, p_r,\gamma_r ; p'_1,\gamma'_1,\ldots,
\widehat{p'_i,\gamma'_i},\ldots, p'_{s+1},\gamma'_{s+1})
\end{eqnarray*}
Let us recall that $\Phi^{(r,s)}$ is antisymmetric in the
electron and the positron variables separately. We have
\begin{equation}\nonumber
\{b_{\gamma,+}(p), b_{\gamma',+}^*(p')\} = \{b_{\gamma,-}(p),
b_{\gamma',-}^*(p')\} = \delta_{\gamma,\gamma'}\delta(p-p')\, .
\end{equation}
Any other anti-commutators equal zero.

\section*{Acknowledgements}
The research of J.-M. B. and J. F. is supported by ANR grant ANR-12-JS0-0008-01. We thank Jean-Fran{\c{c}}ois Bony and Thierry Jecko for helpful discussions.

%%%%%%%%%%%%%%%%%%%%%%%%%%%%%%%%%%%%%%%%%%%%%%%%%%%%%%%%
%%%%%%%%%%%%%%%%%%%%  BIBLIOGRAPHY %%%%%%%%%%%%%%%%%%%%%
%%%%%%%%%%%%%%%%%%%%%%%%%%%%%%%%%%%%%%%%%%%%%%%%%%%%%%%%
\bibliographystyle{amsalpha}

\begin{thebibliography}{A}

\bibitem{ammari} Z. Ammari, \textit{Scattering theory for a class of fermionic models}, J. Funct. Anal. \textbf{208} (2004), 302--359.

\bibitem{ABG} W. Amrein, A. Boutet de Monvel, V. Georgescu, \emph{
  $\mathrm{C}_0$-groups, commutator methods and spectral theory of $N$-body Hamiltonians},
Basel--Boston--Berlin,  Birkh\"auser, 1996.

\bibitem{ABFG} W. Aschbacher, J.-M. Barbaroux, J. Faupin, J.-C. Guillot,
\textit{Spectral theory for a mathematical model of the weak interaction: The decay of the intermediate vector bosons W+/-, II},
Ann. Henri Poincar\'e  \textbf{12}, no.8 (2011) 1539--1570.


\bibitem{bg3} J.-M. Barbaroux, J.-C.~Guillot,
\textit{Spectral theory for a  mathematical model of the weak interactions: The decay of the intermediate vector bosons W+/-}, arXiv:0904.3171.

\bibitem{bg4} J.-M. Barbaroux, J.-C.~Guillot,
\textit{Spectral theory for a  mathematical model of the weak interactions: The decay of the intermediate vector bosons W+/-}, Advances in Mathematical Physics 2009.

\bibitem{BDG} J.-M. Barbaroux, M. Dimassi, J.-C.~Guillot,
\textit{Quantum electrodynamics of relativistic bound states with
cutoffs II.}, Mathematical Results in Quantum Mechanics,
Contemporary Mathematics \textbf{307} (2002), 9--14.

\bibitem{BDG2} J.-M. Barbaroux, M. Dimassi, J.-C.~Guillot,
\textit{Quantum electrodynamics of relativistic bound states with
cutoffs}, Journal of Hyperbolic Differential Equations, \textbf{1}(2)
(2004), 271-314.


\bibitem{Da} E.B. Davies, \emph{Linear operators and their spectra}, Cambridge University Press, Cambridge, 2007.

\bibitem{DG} J. Derezi\'nski,  C. G\'erard, \textit{Asymptotic completeness
in quantum field theory. Massive Pauli-Fierz Hamiltonians}, Rev. Math. Phys.
\textbf{11} (1999), no. 4, 383--450.

\bibitem{FMS} J. Faupin, J.S. M{\o}ller and E. Skibsted, \textit{Second Order Perturbation Theory for Embedded Eigenvalues}, Comm. Math. Phys., \textbf{306}, (2011), 193--228.

\bibitem{GlimmJaffe} J.~Glimm, A.~Jaffe, \textit{Quantum field theory and
statistical mechanics}. Birkhäuser Boston Inc., Boston, MA, 1985.
Expositions, Reprint of articles published 1969--1977.

\bibitem{Georgescu}
V. Georgescu, \emph{On the spectral analysis of quantum field Hamiltonians}, J. Funct. Anal. \textbf{245}, (2007), 89--143.

\bibitem{GGM1}
V. Georgescu, C. G{\'e}rard, J.S. M\o ller, \emph{Commutators,
  $\mathrm{C}_0$--semigroups and resolvent estimates}, J. Funct. Anal., \textbf{216},  (2004), 303--361.

\bibitem{GGM2}
V. Georgescu, C. G{\'e}rard, J.S. M\o ller, \emph{Spectral theory of
  massless Pauli-Fierz models}, Comm. Math. Phys., \textbf{249},  (2004), 29--78.

\bibitem{GP}
C. G{\'e}rard, A. Panati, \emph{Spectral and scattering theory for some abstract QFT Hamiltonians}, Rev. Math. Phys, \textbf{21}, (2009), 373--437.

\bibitem{Go} S. Gol{\'e}nia, \emph{Positive commutators, Fermi Golden Rule and the spectrum of 0 temperature Pauli-Fierz Hamiltonians}, J. Funct. Anal., \textbf{256}, (2009), 2587--2620.

\bibitem{Greiner} W. Greiner, \textit{Relativistic quantum mechanics. Wave equations}, Springer, 2000.

\bibitem{GreinerMuller} W. Greiner, B. M\" uller,
\textit{Gauge Theory of Weak Interactions}, 3rd edition, Springer, 2000.

\bibitem{HuSp} M. H{\"u}bner, H. Spohn, \emph{Spectral properties of the spin-boson Hamiltonian}, Ann. Inst. Henri Poincar{\'e}, \textbf{62}, (1995), 289--323.

\bibitem{IZ} C. Itzykson, J.-B. Zuber, \textit{Quantum Field theory}, Mc
Graw-Hill, New York, (1985).

\bibitem{ref10} D. Kastler, \textit{Introduction \`a
l'Electrodynamique Quantique}, Dunod, Paris, 1960.

\bibitem{Kato} T. Kato, \textit{Perturbation Theory for Linear Operators},
Springer-Verlag, New York Inc. (1966).

\bibitem{Mo}
{\'E}. Mourre, \emph{Absence of singular continuous spectrum
for certain
  selfadjoint operators}, Comm. Math. Phys., \textbf{78}, (1980/81),
  391--408.

\bibitem{RS} M. Reed, B. Simon, \textit{Methods of modern mathematical
physics}, Vol. I and II, Academic Press, New York, 1972.

\bibitem{ref9} M.E. Rose, \textit{Relativistic Electron Theory},
Wiley, 1961.

\bibitem{Schweber} S. Schweber, \textit{An Introduction to
Relativistic Quantum Field Theory}, Harper and Ross, New York, (1961).

\bibitem{Sk}
 E. Skibsted, \emph{Spectral analysis of $N$-body systems coupled to
   a bosonic field}, Rev. Math. Phys.,  \textbf{10},
 (1998), 989--1026.

\bibitem{Takaesu2009}
\newblock T.~Takaesu.
\newblock On the spectral analysis of quantum electrodynamics with
spatial cutoffs. I,
\newblock  {\em J. Math. Phys.}50(2009)06230 .

\bibitem{ref8} B. Thaller, \textit{The Dirac Equation.} Texts and
Monographs in Physics. Springer Verlag, Berlin, 1 edition, 1992.


\bibitem{Tr} H.~Triebel, \textit{Interpolation theory, function spaces, differential operators}, VEB Deutscher Verlag der Wissenschaften, Berlin, 1978.

\bibitem{W} S.~Weinberg, \textit{The Quantum Theory of Fields}
Vol. I., Cambridge University Press, 2005.

\bibitem{W2} S.~Weinberg, \textit{The Quantum theory of fields.}
Vol. II., Cambridge University Press, Cambridge, 2005.

\end{thebibliography}

\end{document}